\documentclass[conference]{IEEEtran}

\def\BibTeX{{\rm B\kern-.05em{\sc i\kern-.025em b}\kern-.08emT\kern-.1667em\lower.7ex\hbox{E}\kern-.125emX}}

  	\usepackage{graphicx}
        \usepackage{svg}
	\usepackage{balance}
	\usepackage{amssymb, amsmath}
    \usepackage{amsthm}
	\usepackage{multirow, rotating, wasysym, url}
	\usepackage[ruled,linesnumbered,vlined]{algorithm2e}
	\usepackage{hyperref} 
	\usepackage{color}
    \usepackage{multicol}
    \usepackage{textcomp}
    \usepackage{epstopdf}
    \usepackage{graphicx}
    \usepackage{listings}
    \usepackage{comment}
    \usepackage{caption}
    \usepackage{enumitem}
    \usepackage{subcaption}
    \usepackage{tabularx}

\graphicspath{{./DrawingFigures/},{./EvaluateFigures/}}
    \newcommand{\nop}[1]{}
	\newcommand{\presec}{\vspace{-0.0in}}
	\newcommand{\postsec}{\vspace{-0.0in}}
	\newcommand{\presub}{\vspace{-0.00in}}
	\newcommand{\postsub}{\vspace{-0.00in}}
	\newcommand{\precaption}{\vspace{-0.05in}}

        \newcommand{\dsname}{ResidualSketch}

	\definecolor{greener}{RGB}{0,166,0}
	\definecolor{reder}{RGB}{255,0,0}
	\definecolor{bluer}{RGB}{0,0,255}
	
\newtheorem{Def}{Definition}

\newtheorem{Lem}{Theorem}
\begin{document}

%
\title{\dsname: Enhancing Layer Efficiency and Error Reduction in Hierarchical Heavy Hitter Detection with ResNet Innovations}

\author{

    Xilai Liu\IEEEauthorrefmark{1}\IEEEauthorrefmark{4},
    Yuxuan Tian\IEEEauthorrefmark{2},
    Xiangyuan Wang\IEEEauthorrefmark{2},
    Yuhan Wu\IEEEauthorrefmark{2},
    Wenhao Wu\IEEEauthorrefmark{1}\IEEEauthorrefmark{4}, 
    Tong Yang\IEEEauthorrefmark{2}, 
    and Gaogang Xie\IEEEauthorrefmark{3}\IEEEauthorrefmark{4}
    \\

    \IEEEauthorblockA{
    \IEEEauthorrefmark{1}\textit{Institute of Computing Technology, Chinese Academy of Science,} \textit{Beijing, China,}\\
    \IEEEauthorrefmark{2}\textit{School of Computer Science, Peking University,} \textit{Beijing, China,}\\
    \IEEEauthorrefmark{3}\textit{CNIC, Chinese Academy of Sciences, China}\\
    \IEEEauthorrefmark{4}\textit{University of Chinese Academy of Sciences, China}\\
    }
    }

\maketitle

%
\begin{abstract}
In network management, swiftly and accurately identifying traffic anomalies, including Distributed Denial-of-Service (DDoS) attacks and unexpected network disruptions, is essential for network stability and security. Key to this process is the detection of Hierarchical Heavy Hitters (HHH), which significantly aids in the management of high-speed IP traffic. This study introduces ResidualSketch, a novel algorithm for HHH detection in hierarchical traffic analysis. ResidualSketch distinguishes itself by incorporating Residual Blocks and Residual Connections at crucial layers within the IP hierarchy, thus mitigating the Gradual Error Diffusion (GED) phenomenon in previous methods and reducing memory overhead while maintaining low update latency. Through comprehensive experiments on various datasets, we demonstrate that ResidualSketch outperforms existing state-of-the-art solutions in terms of accuracy and update speed across multiple layers of the network hierarchy. All related codes of ResidualSketch are open-source at GitHub.
\end{abstract}

%
%

%



\presec
\vspace{-0.1in}
\section{Introduction}\label{sec:introduction}
\postsec
\label{sec:algorithm}

Network measurement plays a critical role in modern network management, particularly in the detection and analysis of traffic anomalies \cite{anomaly2021burst,aoc2024}. 
Such anomalies, ranging from Distributed Denial-of-Service (DDoS) attacks \cite{ddos2024cloud,poseidon,ddos2021jaqen} to unexpected network blocks \cite{blockingIMC,cloudflare_blocking}, where traffic from multiple IP sources may suddenly appear or disappear, pose significant challenges to network stability and security. 
The rapid and precise identification of these irregularities is vital for maintaining network integrity and performance in large-scale network operations.

A key task in network measurement is the identification of Hierarchical Heavy Hitters (HHHs) \cite{online2004}. IP traffic is inherently hierarchical and can be organized in various forms, such as source IP address prefixes or source-destination IP address combinations.  HHHs are defined as traffic aggregates that significantly exceed typical volumes within these hierarchies. Unlike the traditional task of identifying heavy hitters (HH) \cite{heavyhitter}, which concentrates on large individual flows, HHH detection presents a more complex challenge by additionally identifying aggregated smaller flows that cumulatively generate significant traffic at higher prefix levels, commonly referred to as HHs with shorter prefixes. The detection of HHHs is pivotal in network measurement, given that anomalies---ranging from DDoS attacks to unexpected network blocks---manifest as significant traffic fluctuations, often concentrating at specific nodes within the hierarchy.

%
%
%
%

IP traffic, characterized by its high-speed data streams ranging from Gigabits to Terabits per second \cite{DeAR2023}, presents significant challenges in real-time processing due to constraints in processing time and memory. 
Accurate and low-latency monitoring of this traffic, especially under resource limitations, necessitates the deployment of efficient algorithms for tracking various flow events, including HH \cite{heavyhitter}, heavy change \cite{univmon} and \textit{etc.}
Sketched-based approaches \cite{countsketch, cmsketch,cusketch} have gained prominence in this context, with recent advancements \cite{univmon,elasticsketch} focusing on optimizing their design to meet the nuanced demands of network operators.

%
%

\vspace{-0.15in}

\begin{figure}[htbp]
\begin{minipage}[b]{0.48\textwidth}
\centering
\includegraphics[width=\textwidth]{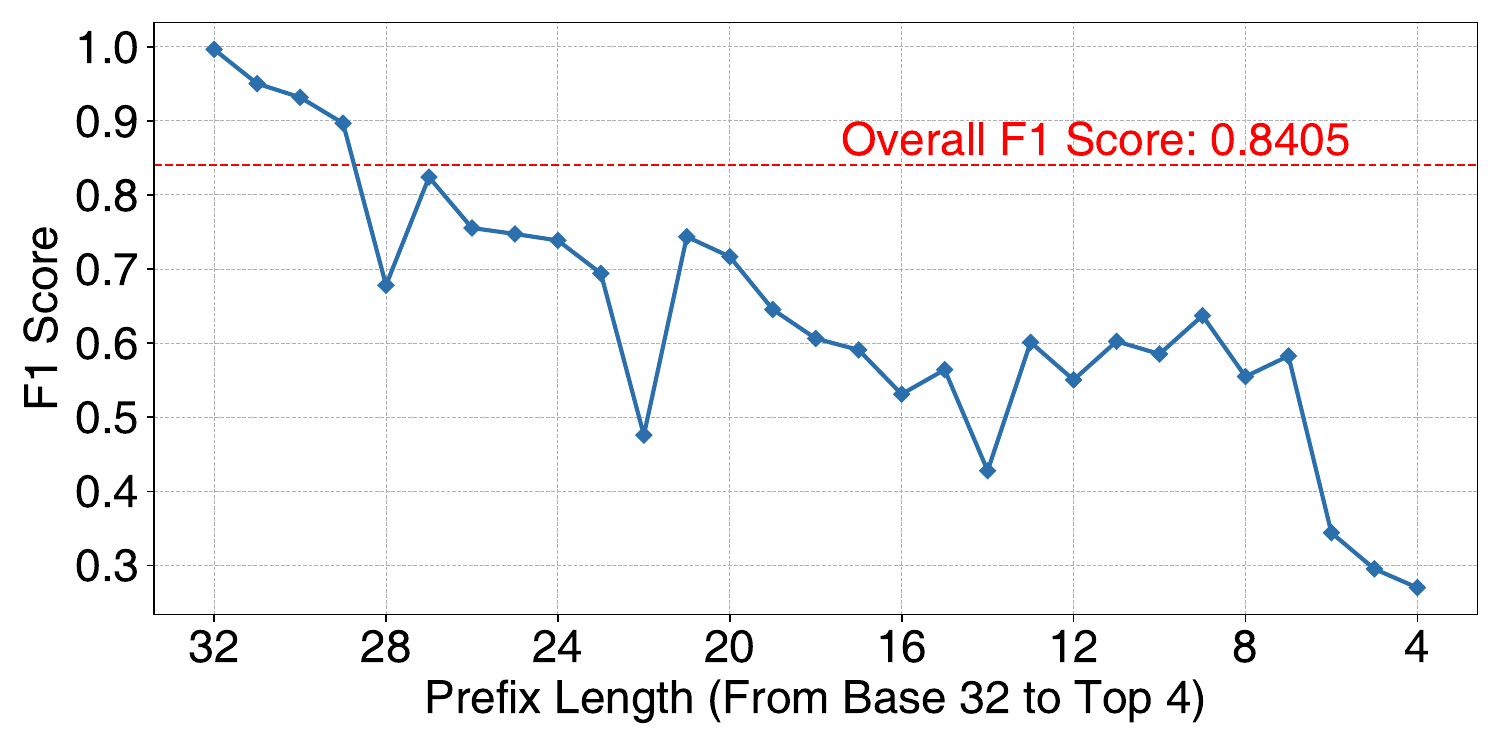}
\end{minipage}
\caption{F1 Score for HHH Detection Across Various Layers in the MAWI Dataset \cite{mawi} Using Cocosketch (100KB)}\label{HHH_layer_error}
\vspace{-0.1in}
\end{figure}

The detection of HHH has been extensively explored, primarily through two methodologies: Heavy Hitter (HH)-based \cite{ancestry,hhh2012,mvpipe} and Arbitrary Partial Key (APK)-based \cite{unbiasedspacesaving, cocosketch, hyperuss} approaches. 
HH-based sketches concentrate on pinpointing heavy flows at each IP layer, employing distinct HH sketches for each layer. 
This process entails the identification of potential HHs across layers, culminating in the formation of the final HHH. 
Subsequent optimizations have enhanced these solutions in terms of accuracy, processing speed, and convergence \cite{RHHH, mvpipe}. 
However, a notable drawback of these methods is the inefficient use of memory, as multiple layers might redundantly record identical HH information.

 On the other hand, APK-based solutions implement a unified sketch methodology, analogous to a singular HH sketch. 
 This method involves the records of singular IP flows to detect heavy flows across all hierarchical layers. 
 A notable drawback of this strategy is the propagation of errors. Particularly, the CocoSketch technique \cite{cocosketch}, an evolved version of APK-based solutions, exhibits a phenomenon of \textit{Gradual Error Diffusion} (GED) as the data traverses through the layers of the IP hierarchy. 
 Unlike the abrupt and substantial error increase seen in gradient explosion phenomena within neural networks, GED manifests as a steady and more gradual decline in accuracy from individual addresses to broader subnets, as shown in Figure \ref{HHH_layer_error}.
 This gradual decline in accuracy, though not as rapid as gradient explosion, becomes increasingly problematic in higher hierarchical layers, especially within extensive subnets or aggregated IP ranges. 
 For in-depth definitions and examinations of above methodologies, refer to Sections \ref{sec:heavyhitter} and \ref{sec:subsetsum}.

In this study, we present ResidualSketch, a sketch-based flow measurement algorithm designed for HHH problem.
By integrating multiple pivotal layers within the IP hierarchy---ranging from individual IP addresses to aggregated subnets---and incorporating residual connection techniques, we can effectively address HHH problem. 
Compared to HH-based Sketches, ResidualSketch achieves provable accuracy guarantees with minimal memory overhead. 
Against APK-based solutions, ResidualSketch demonstrates reduced error rates in estimating subnet activity, all while maintaining comparable update latency.

The primary challenge of ResidualSketch centers on mitigating GED effectively across hierarchical layers of IP addresses and subnets.
To address this, we propose two main techniques. 
(a) Leveraging the identified Heavy Hitter (HH) clustering phenomenon and employing \textit{Residual Blocks}, this approach ensures each \textit{Residual Block} oversees prefix length from one critical layer to the next within the IP hierarchy. This strategy effectively mitigates error propagation at each critical layer, thereby improving the system's overall accuracy.
(b) To remove duplicate counts across different Residual Blocks, we introduce the concept of \textit{Residual Connection}. This technique, inspired by ResNet architectures \cite{resnet}, ensures that counts of heavy flows from lower blocks are not carried into higher blocks. 
Instead, only the \textit{residual} counts---that is, the counts that are additional or unique to each higher block—are maintained. Our key contributions are as follows:

1) We elucidate the structure of HHH and GED phenomenon in state-of-the-art techniques, highlighting HH clustering and error propagation in IP hierarchies.

2) We introduce ResidualSketch, an efficient algorithm designed to address hierarchical IP characteristics, thereby effectively solve HHH detection problem.

3) We theoretically analyze ResidualSketch' properties and conduct extensive experiments on various datasets to show that the ResidualSketch achieves high accuracy and fast update than the state-of-the-art.


\presec
\section{Background and Motivation}
\label{sec:background}
\postsec

In this section, we begin by introducing essential definitions for HH and HHH. Next, we delve into work related to them, including HH-based and APK-based HHH sketch solutions. The frequently used symbol is in
Table \ref{NT}. 

\begin{table}[htbp]
\centering
\vspace{-0.05in}
\caption{Notation Table}
\label{NT}
\begin{tabularx}{0.48\textwidth}{|c|X|}
\hline
Notation & Meaning \\
\hline
$v_f$ & An flow $f$'s value \\
\hline
$N$ & Data stream's length, i.e., packet counts \\
\hline
$\theta$ & A threshold parameter \\
\hline
$L$ & The number of levels \\
\hline
$d$ & The number of layers \\
\hline
$[l_i,l_{i+1})$ & The layer range covered by one level ($l_{L+1}=d+1$) \\
\hline
$B_i$ & $i$-th level's bucket array \\
\hline
$h_{ij}(\cdot)$ & $i$-th level's $j$-th hash function \\
\hline
$B_i[j][m]$ & $m$-th bucket at $i$-th level, $j$-th hashed array \\
\hline
\( P_f \) & Probability of key replacement for new flow $f$ \\
\hline
\end{tabularx}
\vspace{-0.1in}
\end{table}

\presub
\subsection{Problem Definition}
\postsub

In accordance with the conceptual frameworks described in \cite{mvpipe} and \cite{RHHH}, we define data streams as a continual sequence of items, each denoted by a tuple $( f, v_f) $ and processed precisely once. Within the context of network measurement, $f$ acts as a flow identifier, while $v_f$ can either be a unit count or represent packet size. For simplicity, we discretize the unending data streams into finite time windows, focusing on a single window where $N$ denotes the whole packet counts in it. Then, we formally define HH.

\begin{Def}
\label{def:first}
\textbf{Heavy Hitter (HH)}: An item $e$ is an HH if its frequency $v_e$ is above the threshold: $\theta \cdot N$.  The item $e$ can vary in prefix length, but is often a fully specified IP address.
\end{Def}

This work centers primarily on hierarchically structured data streams characterized by flow IDs, akin to IP traffic. For the one-dimensional (1D) HHH problem, $f$ can represent either a source or destination IP address. In the case of the two-dimensional (2D) HHH problem, $f$ encapsulates a source-destination IP address pair. The hierarchical structure can be aggregated at different levels of granularity, either at the byte-level or bit-level, resulting in four distinct types of hierarchies: 1D-byte, 2D-byte, 1D-bit, and 2D-bit. 

\begin{figure}
\centering
\includegraphics[width=0.45\textwidth]{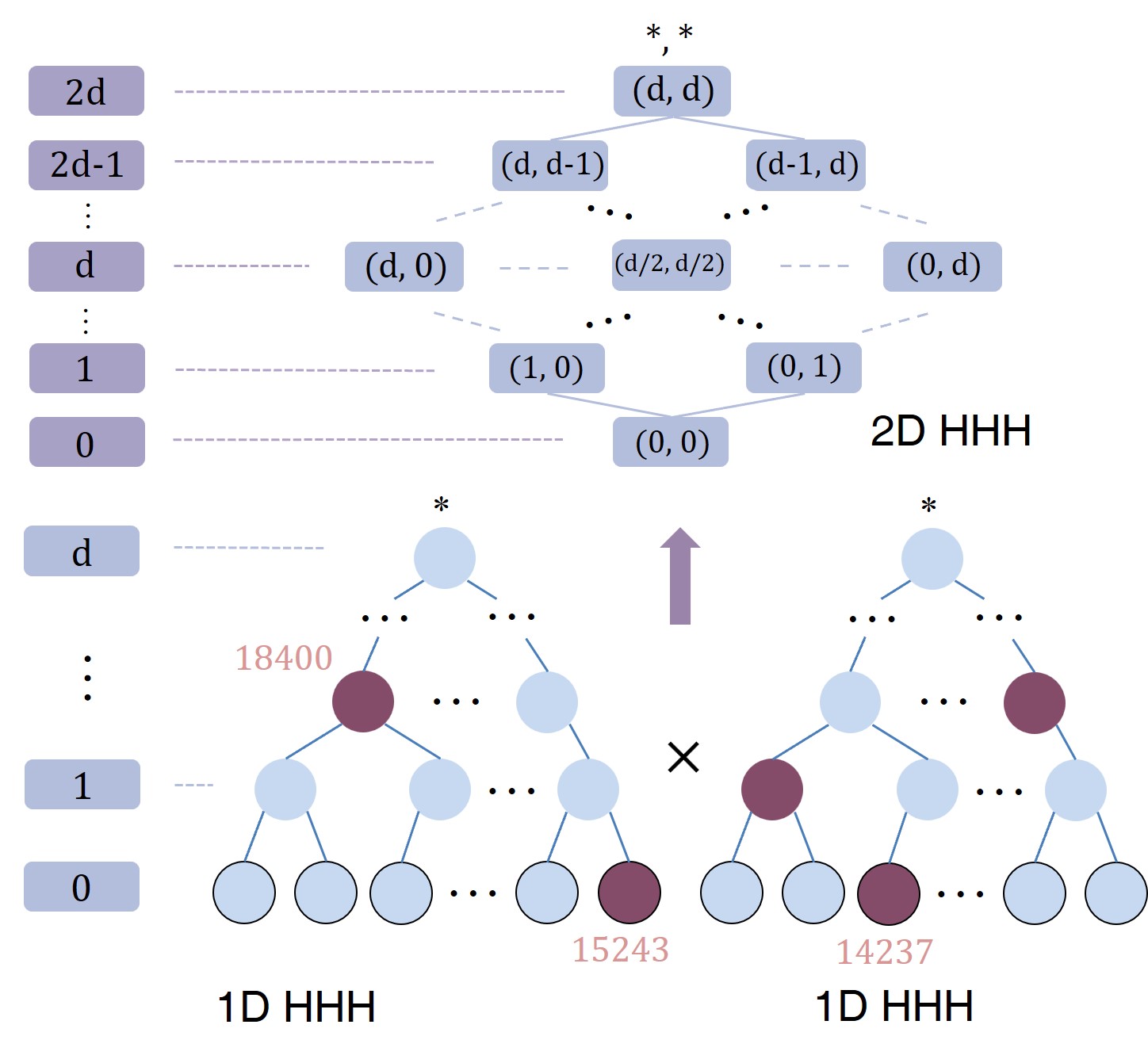}
\precaption
\captionsetup{font=footnotesize}
\caption{The binary tree represents the hierarchical flow IDs structure, such as srcIP or dstIP, aggregating from layer 0 to \(d\), from fully specified IPs to fully generalized '*'. In all figures, dark mauve nodes denote HHs. In this figure, they are shown with \(T\) values (threshold $10k$ here), are counted from the bottom up to form a set known as HHH. In a 2D scenario, there are \((d + 1)^2\) nodes and \(2d + 1\) layers of aggregation for HHH determination. }
\label{gp}
\vspace{-0.2in}
\end{figure}

As Figure \ref{gp} shows, a node's layer in a hierarchical structure corresponds to its depth, ranging from $0$ at the base to $d$ at the apex, with the layer index inversely related to prefix length. Here, $d$ can be $4$ or $32$ according to byte or bit granularity. At layer 0, domain $\mathcal{U}$ is composed of fully specified IP addresses ready for generalization. For example, the specific IP address $p$, denoted as $192.168.0.1$, can be generalized to its prefix $q$, represented as $192.168.0.*$. Here, prefix $q$ is the ancestor of $p$, making $p$ a descendant of $q$. This generalization process establishes the hierarchical layers, depicted as a tree structure.

For the HHH detection, our primary objective is to locate all HHs at various hierarchical layers. Notably, discovering an HH \( p \) inherently implies that its ancestral nodes may also be HHs, leading to duplicate entries. To address the issue of duplication, we introduce \textit{conditional count}, denoted by \( C_{p \mid P} \). This metric computes the frequency of \( p \) after discounting the frequencies of previously discovered HHs listed in the set \( P \). We employ a bottom-up approach, starting from the base layer (0) and proceeding to the higher layers to preclude duplicate counting. Initially, we scrutinize fully specified flows and accumulate their HHs into the set \( HHH_0 \). Utilizing \( HHH_0 \), we calculate the conditional counts for flows at level 1. If \( C_{p \mid HHH_0} \geq \theta \cdot N \), we incorporate \( p \) into the set \( HHH_1 \). We continue this methodology until the highest level is reached. The HHH can be formally defined as follows. 

\begin{Def}
\label{def:two}
\textbf{Hierarchical Heavy Hitter (HHH)}: The set $HHH_0$ contains the HHs of fully specified flows. For the HHH set at each layer $l$, $HHH_l$, $0< l \leq d$, we define:
$$
H H H_l=H H H_{l-1} \cup\left\{p:\left(p \in \operatorname{layer}(l) \wedge C_{p \mid H H H_{l-1}} \geq \theta  N\right)\right\} .
$$
The cumulative set $HHH_d$, encompassing all HHs, is identified as the primary target for detection. 
\end{Def}

\presub
\subsection{HH-based HHH sketches}
\postsub
\label{sec:heavyhitter}

Notable algorithms in HH detection, such as Elastic Sketch \cite{elasticsketch}, SpaceSaving \cite{spacesaving}, Lossy Counting (LC) \cite{lossycounting}, MV Sketch\cite{mvsketch}, and Unbiased Space Saving (USS) \cite{unbiasedspacesaving}, employ limited counters to capture frequent data flows while filtering out less frequent ones. Both Space-Saving and USS employ a Stream-Summary data structure, with USS incorporating probabilistic mechanisms for unbiased results, commonly categorized within the APK-based sketch. Meanwhile, MV Sketch \cite{mvsketch} distinguishes candidate HHs through majority voting.

Building upon these methods, HHH algorithms typically maintain an HH sketch instance per layer and incorporate specialized enhancements. SpaceSaving, for example, has been adapted in studies like \cite{hhh2012, RHHH} for HHH-specific detection. Based on LC method, Cormode's team present an algorithm termed \textit{full ancestry} \cite{ancestry} that utilizes a dynamic lattice structure for node addition and removal. MVPipe, an adaptation of MVSketch, can be even deployed in programmable switches but is limited to 1D-byte HHH detection. Nevertheless, these approaches face limitations due to excessive allocation of HH sketch instances across layers, leading to memory inefficiency, particularly in layers with sparse HH counts. 

\presub
\subsection{APK-based HHH sketches}
\postsub
\label{sec:subsetsum}

The task of Arbitrary Partial Key (APK) querying, as delineated in \cite{cocosketch}, involves the retrieval of arbitrary partial keys corresponding to a predefined full key. Grounded in the same theoretical framework as the subset sum problem, leading methodologies in this domain, such as those presented in \cite{unbiasedspacesaving}, \cite{cocosketch} and \cite{hyperuss}, achieve high fidelity in emulating Probability Proportional to Size (PPS) sampling. Consequently, these approaches are able to preserve a greater extent of information relative to prior sketch-based techniques for addressing the APK task. These methods are also naturally capable of addressing HHH problem. In this context, a full key constitutes the complete srcIP or dstIP, while a partial key represents any IP address prefix of variable length. 

The limitation of APK-based HHH sketches stems from their GED characteristic. For example, the USS employs unbiased PPS sampling for estimation. As layer depth increases, an accumulation of smaller flows contributes to the formation of an HH, exacerbating the sampling error. Nevertheless, this issue does not significantly impair the overall accuracy metrics, as a large fraction of HHs are identified at the base layer mitigating this drawback. This phenomenon is corroborated by Figure \ref{HHH_layer_error}, where the overall F1 score tends towards the base layer's F1 score, yielding better results than expected. However, this methodology contradicts the primary goal of certain applications, such as DDoS mitigation, which require accurately identifying the aggregation points of smaller flows---a task at which this method is inadequate.

This motivates the design of ResidualSketch, a combination of HH- and APK-based HHH sketches. It employs Residual Blocks to mitigate sampling error in upper layers and Residual Connections to minimize redundancy across Residual Blocks.

\presec
\section{\dsname}\label{sec:problemdef} \postsec

\subsection{Design} \postsub

\noindent\textbf{Data Structure:} As Figure \ref{insertion_1} shows, the data structure is organized into \( L \) pivotal \textit{levels}, each of which is a bucket array corresponding to a specific pivotal layer in the IP tree structure. Each level is responsible for overseeing the layers in the range \([l_i, l_{i+1})\), where \( l_{L+1} = d + 1 \). These bucket arrays at each level are also  termed as Residual Blocks. At each \( i \)-th level (or Residual Block), the bucket array is partitioned into \( g \) hashed arrays using \( g \) unique hash functions. These functions are denoted as \( h_{ij}(\cdot) \) for \( 1 \leq j \leq g \). Each hashed array contains \( b_i \) buckets, where each bucket is structured to store a pair consisting of a key $k$ and its corresponding attribute value $v$. To uniquely identify a specific bucket, the \( m \)-th bucket in the \( j \)-th hashed array at the \( i \)-th level is labeled as \( B_i[j][m] \), where \( 1 \leq i \leq L \), \( 1 \leq j \leq g \), and \( 1 \leq m \leq b_i \).

\noindent\textbf{Residual Block:} Residual Block should belong to the class of APK-based sketch, including USS, Cocosketch and Hyper-USS. When the Residual Block is USS, $g$ equals $1$ and it uses the Stream Summary technique to identify the specific bucket. For Cocosketch, we should use the version without \textit{Circular Dependency Removal}. Hyper-USS is applicable for scenarios where a single flow is associated with multiple attribute values. The allocation of memory for these algorithms is closely linked to a threshold parameter \(\epsilon\), serving to limit the estimation error for any keys to a maximum of \(\epsilon N\). 
The required initial memory allocation is at least \(\frac{1}{\epsilon}\) to ensure the estimation error remains within the specified bound.

\begin{figure}
\includegraphics[width=0.47\textwidth]{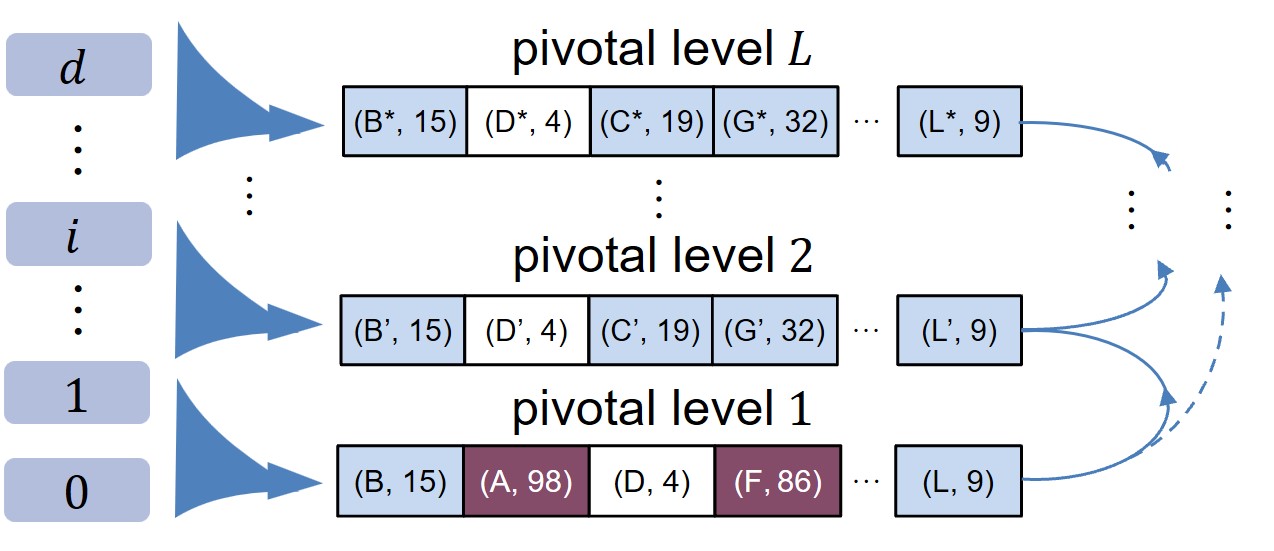}

\caption{The Architecture of Our Algorithm (A*, A' represents A's ancestor)}
\vspace{-0.2in}
\label{insertion_1}
\end{figure}

\noindent\textbf{Insertion:} To insert a flow $f$ into the data structure, it the passes through a pipeline across the levels of bucket arrays. Initially, in the first level, the key \( f \) is hashed into \( g \) distinct buckets---specifically, \( B_1[j][h_{1j}(f)], 1 \leq j \leq g \)---using \( g \) hash functions. When matching this incoming key \( f \) with existing keys \( B_1[j][h_{1j}(f)].k_j \) within these buckets, we encounter one of three scenarios:

\textit{Case 1:} The key \( f \) is already present in at least one of the \( g \) hashed buckets \( B_1[j][h_{1j}(f)] \). In this situation, we directly update its value \( v_f \) in the corresponding bucket.

\textit{Case 2:} The key \( f \) is not found in any of the \( g \) hashed buckets, and at least one bucket is vacant. Here, we select an empty bucket to store the incoming item.

\textit{Case 3:} The key \( f \) is absent in all \( g \) hashed buckets, and each bucket is already occupied. In this instance, we identify the bucket with the smallest attribute value. We use emulated PPS sampling technique to maintain unbiasness, which is the core part of APK-based sketches. The new incoming item \( f \) is competed with the item in this smallest bucket based on value frequency. The probability of successfully replacing the existing key $k_j$ with \( f \),  \( P_f \), is given by \( \frac{v_f}{v_f + \min_{j} B_1[j][h_{1j}(f)].v} \). The pseudo-code is shown in Algorithm \ref{insertion_algo}. 

\noindent\textbf{Residual Connection:} After inserting a flow item \( (f, v_f) \) into the bucket arrays at the first level, we immediately get its attribute value. Should this value exceed a predefined threshold---a parameter closely tied to the underlying Residual Block's---two specific actions are initiated: 

1. The item is marked to avoid eviction at this current bucket array level, essentially securing its position. 

2. The item is tagged to prevent it from being recorded again in the bucket arrays of higher levels. When it just equals the threshold, we locate the bucket in the higher levels and subtract its value and it's not recorded anymore.

If the attribute value fails to exceed the predefined threshold after insertion at the first level, flow \( (f, v_f) \) is forwarded to the subsequent level for continued insertion. This concurrent check-and-forward mechanism ensures that only items with attribute values surpassing the threshold are retained at their initial insertion level, while also preventing their duplication in upper levels. For the predefined threshold, it is usually set to a moderate size to prevent premature or delayed invalid locking. Therefore, it is recommended to set the threshold slightly lower than the general threshold setting of HH. We term this technique of allowing only the residual count of HHs to be transmitted to the next levels as Residual Connection. 

\begin{figure}
\includegraphics[width=0.47\textwidth]{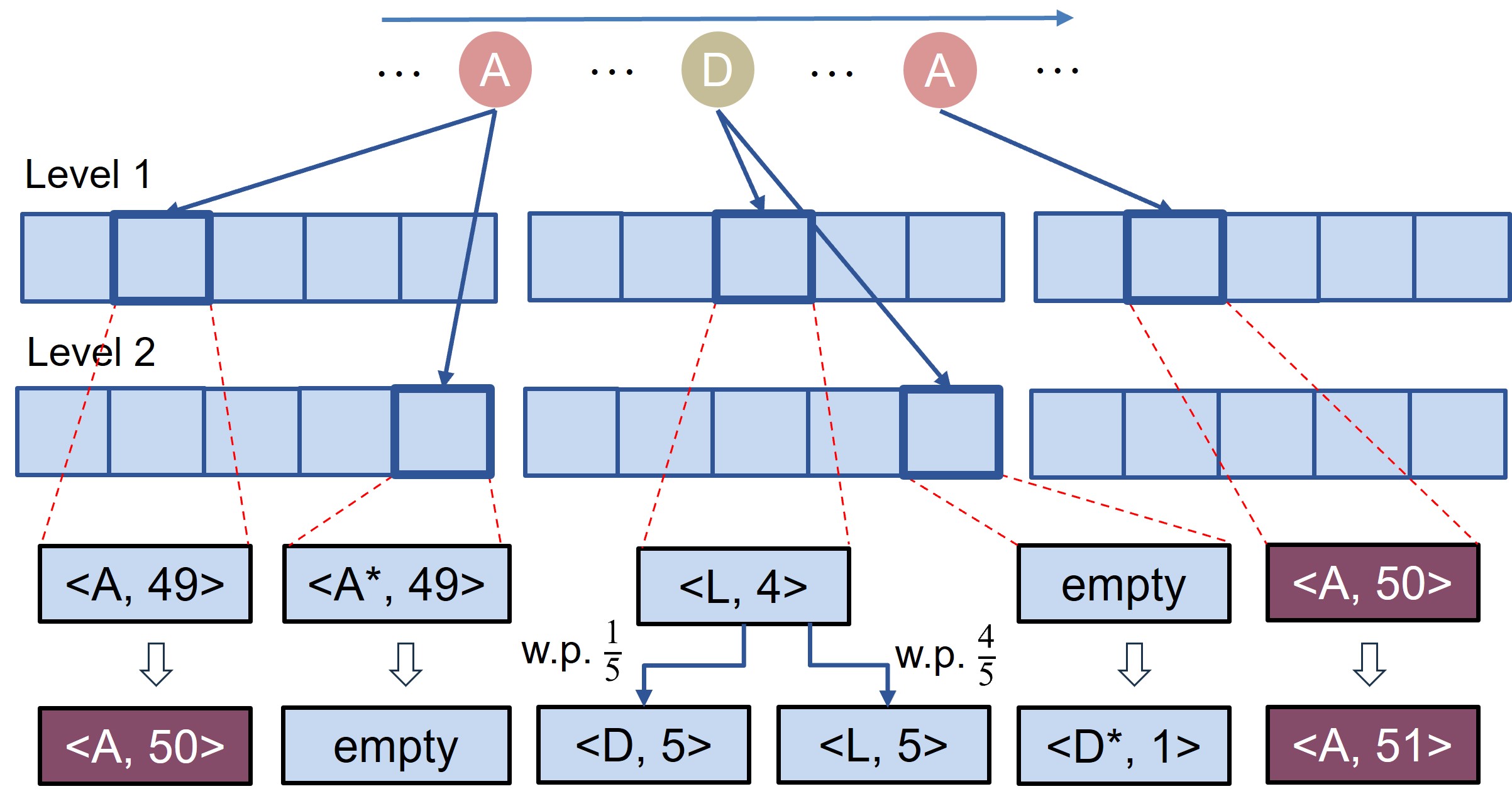}

\caption{Insertion Examples of Residual Sketch (USS as the Residual Block)}
\label{insertion_2}
\vspace{-0.2in}
\end{figure}

\begin{algorithm}[htbp]
\caption{Insertion Procedure}\label{insertion_algo}
\DontPrintSemicolon  

\KwIn{an incoming flow $f$ with value $v_f$}
\SetKwFunction{FMain}{Insert\_level\_i}
\SetKwProg{Fn}{Function}{:}{}
\Fn{\FMain{$f$, $i$, $isDecrement$}}{
    $Min \gets \infty, MinPos \gets 0$\;
    \For{$j \gets 1$ \KwTo $g$}{
        \If{$f == B_i[j][h_{i,j}(f)].k_j$}{
            $B_i[j][h_{i,j}(f)].v \mathrel{{+}\mspace{-1mu}{=}} v_f$\;
            \If{$i \neq 1$ \textnormal{\textbf{and}} $isDecrement$}{
                $B_i[j][h_{i,j}(f)].v \gets max(0, B_i[j][h_{i,j}(f)].v - \theta_i$)\;
            }
            \KwRet{$B_i[j][h_{i,j}(f)].v$}\;
        }
        \If{$B_i[j][h_{i,j}(f)].v < Min$}{
            $Min \gets B_i[j][h_{i,j}(f)].v, MinPos \gets j$\;
        }
    }
    $B_i[MinPos][h_{i,MinPos}(f)].v \mathrel{{+}\mspace{-1mu}{=}} v_f$\;
    \If{$rand() \% B_i[MinPos][h_{i,MinPos}(f)].v < v_f$}{
        $B_i[MinPos][h_{i,MinPos}(f)].k_j = f$\;
    }
    \KwRet{$0$}\;
}
\For{$i \gets 1$ \KwTo $d$}{
    $isDecrement \gets 0$\;
    \If {$i \neq 1$ \textnormal{\textbf{and}} $cnt > \theta_i$}{
        break\;
    }
    \If{$i \neq 1$ \textnormal{\textbf{and}}  $cnt == \theta_i$}{
        $isDecrement \gets 1$\;
    }
    $f = f$ $\& \text{ mask}[l_i]$\;
    $cnt =$ \FMain{$f$, $i$, $isDecrement$}\;
}

\end{algorithm}

\noindent\textbf{Insertion example: } In Figure \ref{insertion_2}, we use USS as the Residual Block and only one bucket is located at each level due to Stream Summary technique. There are only two levels in the example and we choose $50$ as the threshold. We insert item $A$ first. The Stream Summary technique locates the bucket with the key and the recorded attribute is updated. As it surpasses the threshold at first time, we locate the bucket in the second level recording its ancestor $A*$ and delete it, the corresponding bucket becomes empty. Next, we insert the item $D$. As no buckets contains this key and the first level is full, the bucket with the lowest attribute value is located, the replacement probability $P$ is set to $\frac{1}{1+4} = \frac{1}{5}$. As the item doesn't surpasses the threshold, we go to next level. At this time, we find one empty bucket left out by insertion of $A$ and we directly insert to it. In the third insertion, the item $A$ is inserted again. After we located the bucket recording $A$, we update it. As it surpasses the threshold, we don't go to next level and insertion ends.

\noindent\textbf{HHH detection:} At the end of each time window, to estimate the HHH, we begin by sequentially scanning all non-empty buckets across all levels, starting from the bottom. 
In each level, defined by the layer interval $[l_i, l_{i+1})$, we initially restore flow counts to their original values by compensating for any reductions caused by residual connections from lower levels. 
This adjustment yields the active buckets within the $l_i$ layer. 
Subsequently, it computes the sets $HHH_j$ for $j \in [l_i, l_{i+1})$, determining these based on conditional counts in a bottom-up approach across layers, in alignment with the definition of HHH. 
Specifically, buckets within the $j$-th layer with counts exceeding the threshold $\theta N$ undergo a conditional count calculation against $HHH_{j-1}$. 
Buckets whose adjusted counts still surpass $\theta N$ are incorporated into $HHH_j$. 
This step concludes with the integration of $HHH_{j - 1}$ into $HHH_{j}$. 
Buckets not meeting the threshold are merged based on their prefix length, $d - j - 1$, to form buckets for the subsequent layer, $j+1$.
This iterative process is continued until the $l_{i+1}$-th layer is reached. 
Upon completing the calculations across all levels, the set $HHH_d$ is obtained as the final outcome.

\subsection{Level Setting and Memory Allocation} \postsub

The level setting is closely tied to the performance of ResidualSketch. When $L$ is equal to $d$, it becomes HH-based HHH sketch with Residual Connections. When $L$ is equal to 1, it degenerates into APK-based HHH sketch. We should choose an appropriate level number to balance and avoid the deficiency of two of them.

The level number and its corresponding range are selected according to the potential for HHH cluster formation. we hypothesize the existence of three cluster ranges:
\begin{itemize}
    \item A cluster with a layer of prefix length 32, where individual IP addresses are likely to form a significant number of HHs.
    \item A cluster with a layer of prefix length near 24,  characterized by the predominant subnet prefix allocations for small to medium-sized corporations \cite{bgpprefix}.
    \item A cluster, characterized by its relatively short prefix length, ranging from 8 to 12, which aggregates numerous small flows.
\end{itemize}

As illustrated in Figure \ref{HHHratio}, large backbone networks demonstrate attributes that align with the previously described cluster formations. In light of this observation, we recommend adopting either a 2 or 3 level hierarchy. On one hand, allocating a greater number of levels leads to fewer buckets per level, necessitating a larger initial memory allocation to meets the requirements of each layer's HH data structure. On the other hand, such a configuration conforms to the cluster structure inherent in HHH distribution, thereby facilitating optimal performance. 

In allocating memory across levels, priority should be given to ensuring each level has enough memory to handle the HH count before adjustments for Residual Connections. The detail that upper levels track residual HH counts from lower levels, suggesting they might need less memory, is secondary. To simplify, each level is allocated the same amount of memory, ensuring sufficiency and avoiding unnecessary complexity from presumed differences in needs across levels.

\vspace{-0.1in}
\begin{figure}[htbp]
\begin{minipage}[b]{0.48\textwidth}
\centering
\includegraphics[width=\textwidth]{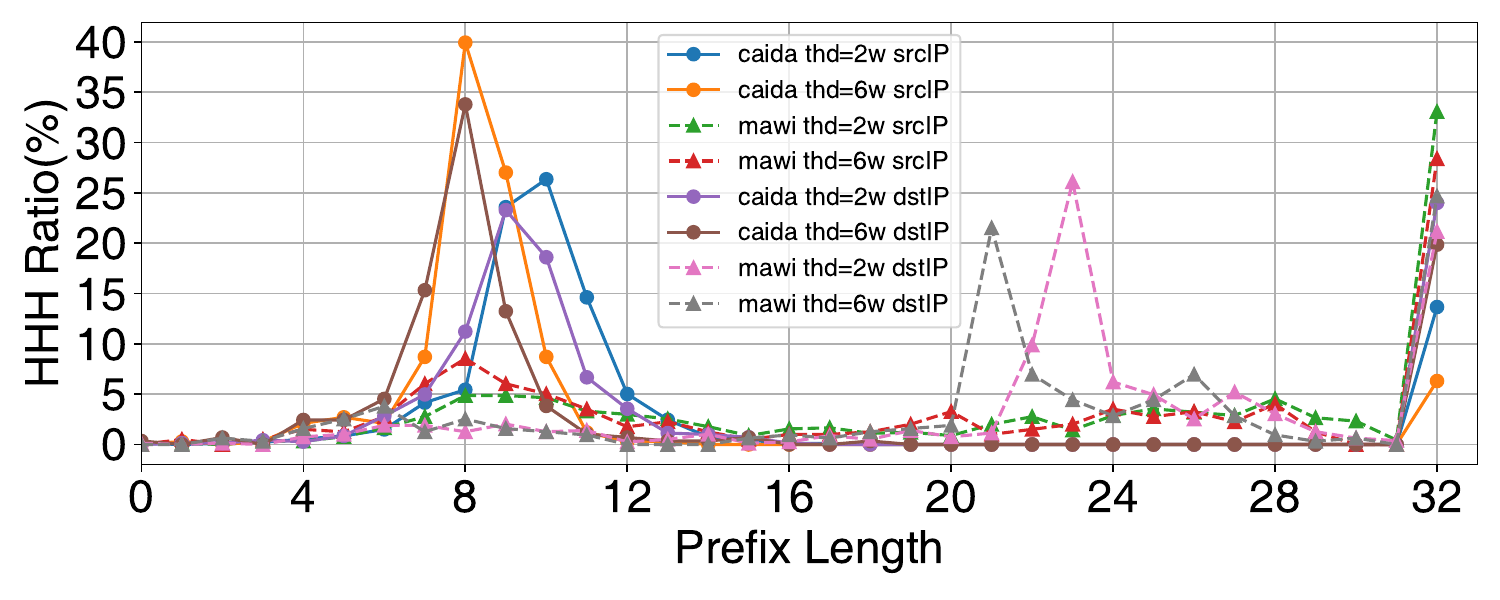}

\end{minipage}
\precaption
\caption{HH ratio of Different Layers on the CAIDA and MAWI Datasets}\label{HHHratio}
\vspace{-0.15in}
\end{figure}

\presub
\subsection{Discussion} \postsub

\noindent\textbf{Residual Connection Strategy: } Duplication at various levels results in higher space occupancy and lower processing throughput, necessitating a method that is both simple and effective. 
Traditional approaches like Full Ancestry \cite{ancestry} mitigate duplication by retaining only frequency differences in upper layers, yet require the maintenance of numerous nodes, thereby increasing space complexity. 
MVPipe \cite{mvpipe} utilizes a majority vote mechanism for efficient duplicate elimination by distributing flow information across layers, though this method still necessitates considerable memory allocation for all layers in advance, leading to inefficiency. 
We posit that duplication primarily occurs in potential large flow records, presenting opportunities for elimination.
The outcome of competition among smaller flows is characterized by unpredictability, making it challenging to foresee potential duplication.
In response, our algorithm introduces the Residual Connection strategy as a simple and effective solution, significantly boosting processing throughput without compromising accuracy.

\noindent\textbf{Two Dimensional HHH Detection:} In 2D HHH detection, the relationship of generalization creates a lattice structure. Within this framework, the APK-based sketches present even more benefits compared to the HH-based sketches. This is due to the requirement of the latter to instantiate a HH instance for every node within the lattice, which is not as efficient. Despite this, the HH clustering phenomenon persists, with flows of interest predominantly observed in these clusters. Constructing levels closer to these clusters renders the ResidualSketch method more advantageous than APK-based sketches. This is because the error propagated to clusters is significantly larger for APK-based sketches, and the accuracy is not as high as levels constructed closer to the clusters. This principle aligns with the rationale behind 1D HHH detection.

\noindent\textbf{Sliding Window: }The technique of supporting a sliding window typically involves the use of smaller sub-windows \cite{slidingsketch,microscopesketch,omniwindow}. In these methods, the counts of flows within these sub-windows are smaller compared with the entire window, leading to more frequent competitions and evictions in an APK-based sketch. This increase in eviction rate can result in larger sampling errors for upper-layer HHs if a single level is utilized exclusively at the base layer. Consequently, the employment of multiple levels becomes necessary, further highlighting the advantages of our algorithm.


\presub
\section{Theoretical Analysis} \postsub
In this section, we first analyze the GED characteristics of APK-based sketches, then detail ResidualSketch’s space and time complexities,  unbiasness and error bounds.

\subsection{Gradual Error Diffusion (GED) of APK-based sketches}
APK-based sketches are inherently unbiased, allowing for a qualitative analysis of error of HH estimation from the perspective of variance. Specifically, GED manifests as an increase in the variance of the HH corresponding to higher layer indexes. 

\begin{figure}[htbp]
\vspace{-0.15in}
\begin{minipage}[b]{0.48\textwidth}
\centering
\includegraphics[width=\textwidth]{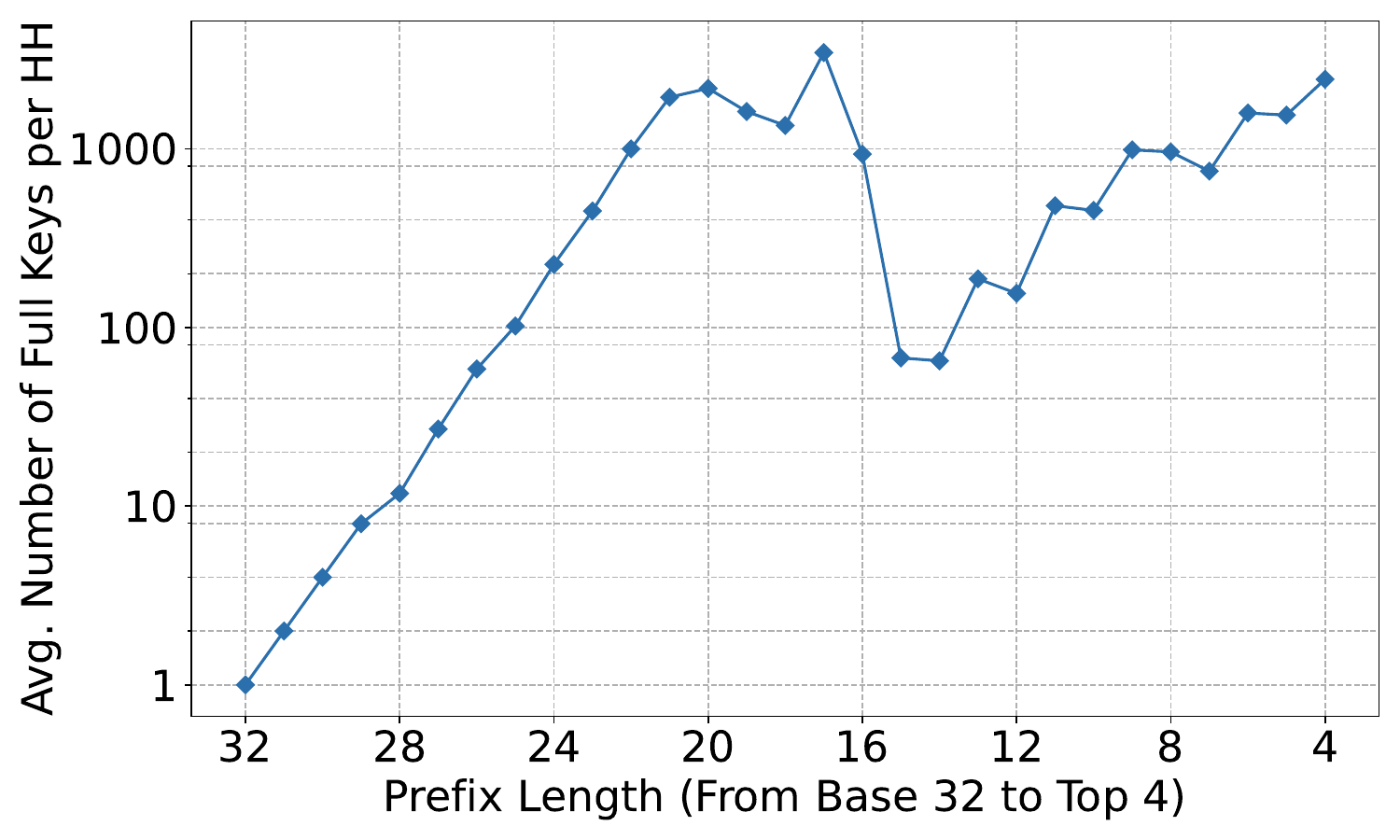}
\end{minipage}
\caption{Average Composition of Full Keys Across Various Layers in the MAWI Dataset \cite{mawi}}\label{HHH_layer_cnt}
\vspace{-0.1in}
\end{figure}

The GED phenomenon emerges due to two key observations. First, the composition of HH becomes increasingly dominated by smaller flows with the advancement of layer index, as evidenced by the average number of full keys constituting each HH in different layers in Figure \ref{HHH_layer_cnt}. 

Second, the variance proof from the USS paper suggests the estimated variance of \(f_p\), \(var(f_p)\), is \(E(N_{min}\kappa_S)\), where \(N_{min}\) is the minimum bucket value, and \(\kappa_S\) is the count of insertions before an element becomes "sticky" in the sketch. As the layer depth increases, the proportion of smaller flows in the HH elements' set increases, leading to a rise in \(\kappa\). Large flows quickly attain "stickiness" in the sketch, resulting in a lower \(\kappa\), whereas HHs with many smaller flows are less likely to become "sticky" immediately, increasing \(\kappa\) with more frequent counts of smaller flows. When \(\kappa\) approaches its peak at the HH frequency, the error margin stabilizes. This stabilization is evidenced by the trend towards a consistent F1 score in subsequent layers, as illustrated in Figure \ref{HHH_layer_error}.

The characteristic of GED is reflection of unfairness in the results for large and small flows in APK-based sketches concerning the HHH problem. Previous research on the Double-anonymous sketch \cite{doubleanonymous} has already identified the exposure of this unfair weakness in the context of the global top-K issue. This also signifies that GED, due to its gradual nature, does not suffer from the complete inaccuracy in estimating small flows as seen with previous frequency estimation methods such as SpaceSaving \cite{spacesaving}. Therefore, for the HHH problem, it is only necessary to restructure the data on critical layers before the error diffusion becomes significantly large.

\vspace{-0.05in}
\subsection{Space and Time Complexities of ResidualSketch}
In 1D HHH detection, the analysis presented in \cite{hhh2012} indicates that the data structure at each level requires $\frac{1}{\epsilon}$ buckets to guarantee a maximum error of $\epsilon N$ in the worst case, where $0<\epsilon<\theta<1$. Consequently, ResidualSketch identifies HHHs utilizing $O(\frac{L}{\epsilon})$ space in worst case. For the time complexity, the update process at each level is nearly equivalent to the update mechanism of the original Residual Block algorithm, aside from the Residual Connection operation. As a result, the update at each level can be carried out in \(O(1)\) time, leading to a worst-case update time of $O(L)$ across all levels.

It is noteworthy that in scenarios involving actual skewed distribution flows, large flows at the lower levels, upon reaching the algorithm's predefined threshold, are not recorded in the upper layers through the Residual Connection and will not be updated in the upper layers with subsequent arrivals. Consequently, the actual time and space consumed are significantly less than the values mentioned above.

\vspace{-0.05in}
\subsection{Unbiasness of ResidualSketch}
\begin{Lem}
ResidualSketch provides an unbiased estimation for any flow key $f_p$, for any prefix length $p$, $E(\hat{V_{f_p}}) = V_{f_p}$.
\end{Lem}
\begin{proof}Residual Sketch segments the prefix length into sequential ranges $(d - l_{i+1}, d - l_i]$ utilizing Residual Blocks. For a flow key of a particular prefix length $p$, its estimate is derived from the corresponding Residual Blocks. 

Initially, we examine scenarios devoid of residual connections, where the prefix length $p$ is precisely $d - l_i$ and corresponds to layer number $l_i$ at level $i$. In such instances, a flow key $(f_p,v_{f_p})$ with prefix length $p$ is treated equivalently to a full key for level $i$. In insertion, if a matching bucket is found, its value increments by $v_{f_p}$, maintaining an unbiased estimate for $f_p$. If no match is found, the smallest bucket, $B^{\text{min}}$, identified as $B_i[\text{min}][h_{i,\text{min}}(f_p)]$, is updated. This bucket's key is replaced with probability $P_{f_p} = \frac{v_{f_p}}{B^{\text{min}}.v+v_{f_p}}$ and its value increases by $v_{f_p}$. The expected increase for the original key in $B^{\text{min}}$ is:
\[
     P_{f_p} \cdot 0 + (1 - P_{f_p}) \cdot (B^{\text{min}}.v + v_{f_p})- B^{\text{min}}.v = 0
\]
The expected increase for key $f_p$ is:
\[
     P_{f_p} (B^{\text{min}}.v + v_{f_p}) + (1 - P_{f_p}) \cdot 0 - 0 = v_{f_p}
\]

Therefore,  after each insertion, the estimation of related keys of prefix length $d - l_i$ are unbiased. 

For prefix lengths engaging residual connections, the key at level $i$ aggregates keys with prefix $d - l_i$ within descendant set $D$. Adding estimations from descendant HHs in set $H$, the expected value $E(\hat{V_{f_p}})$ equals the sum of expected values for $D \cup H$, expressed as $\sum_{e\in D \cup H}E(\hat{V_e})=\sum_{e\in D \cup H}V_e=V_{f_p}$.

\end{proof}

\subsection{Variance Increment of ResidualSketch}
In the context of a single level update, as delineated in the Cocosketch analysis \cite{cocosketch}, each substitution of $(f_p,v_{f_p})$ into $B^{\text{min}}$ results in an incremental variance of $2\cdot v_{f_p}\cdot B^{\text{min}}.v$. The ResidualSketch offers enhancements over APK-based sketches through two mechanisms: aggregation of full keys at a higher level reduces the likelihood of replacements; and the implementation of a Residual Connection leads to a reduction in $B^{\text{min}}.v$, thereby diminishing the variance increment. Consequently, the reduction in variance increment, coupled with the utilization of HH clustering, mitigates the GED phenomenon.



\section{Experimental Results}

\subsection{Experimental Setup}
\noindent \textbf{Platform and setting}: We perform all experiments on a machine with one 12-core processor (AMD EPYC-7662 CPU @ 2.0GHz) and 32G DDR4 memory. The processor has 768KB L1 cache, 6MB L2 cache for each core, and 192MB L3 cache shared by all cores. The OS is Ubuntu 20.04. 

We implement all codes with C++ and build them with g++ 9.4.0 and -O2 option. The hash functions we use are 32-bit Bob Jenkins Hash. The source code for the implementation is publicly available on GitHub \cite{git_src}.

\label{expset}
\noindent \textbf{Datasets}: We have employed four datasets for our study, including three from real-world sources and one synthetic. They are as follows:

\noindent \textbf{CAIDA Dataset}: This dataset includes streams of anonymized backbone traces from high-speed monitors sourced from CAIDA \cite{caida}. It contains approximately 27 million packets, within which there are about $250$k unique flows based on srcIP.

\noindent \textbf{Campus Dataset}: This dataset consists of IP packets gathered from our university's network system. The Campus Dataset contains a total of 30 million packets, distributed among 61,425 distinct flows based on srcIP.

\noindent \textbf{MAWI Dataset}: Provided by the MAWI Working Group \cite{mawi}, this dataset offers real traffic trace data. The MAWI Dataset contains roughly 88 million packets based on srcIP.

\noindent \textbf{Synthetic Dataset}: Utilizing a skewness modification technique from MVPipe \cite{mvpipe}, we replaced the top 1000 IPs from the CAIDA dataset, representing $54\%$ of traffic, with randomly generated IPs. These IPs are engineered to aggregate into HHs at prefix lengths 20 to 24, enabling effective accuracy assessment of competing algorithms.

\noindent \textbf{Metrics}:

\noindent\textbf{1) Precision Rate (PR)}: PR is defined as $\frac{|\Phi \cap \Theta|}{|\Theta|}$. Here $\Theta$ is the set of reported $HHH_d$ and $\Phi$ is the set of correct $HHH_d$.

\noindent\textbf{2) Recall Rate (RR)}: RR is defined as $\frac{|\Phi \cap \Theta|}{|\Phi|}$. Here $\Theta$ is the set of reported $HHH_d$ and $\Phi$ is the set of correct $HHH_d$.

\noindent\textbf{3) F1 Score}: F1 Score is defined as $\frac{2\text{PR}\cdot\text{RR}}{\text{PR} + \text{RR}}$. 

\noindent\textbf{4) Average Relative Error (ARE)}: ARE is defined as $\frac1{|\Psi|}\sum_{e_i\in\Psi}\frac{\left|f_i-\hat{f_i}\right|}{f_i}$
where $f_i$ is the real frequency of HH $e_i$, $\hat{f_i}$ is its estimated frequency of HH $e_i$, and $\Psi$ is the query set.

\noindent\textbf{5) Throughput}: The number of processed packets per second, in million packets per second (Mpps).

\subsection{Comparison with Prior Art}
In this section, due to the significant differences between HH-based sketches and APK-based sketches, we conduct separate comparisons of our algorithm with each, respectively.

\subsubsection{Comparison with HH-based Sketches}
We compare ResidualSketch with several state-of-the-art HHH algorithms, including full ancestry (FULL) \cite{ancestry}, HHH12 \cite{hhh2012}, RHHH \cite{RHHH} and MVPipe \cite{mvpipe} for HH-based sketches on MAWI dataset. 

\noindent\textbf{Memory (Figure \ref{fig:mawibyteMEM}, \ref{fig:mawibitMEM})}: For HHH12 and RHHH, we allocate memory according to their recommended settings, taking into account thresholds and acceptable error tolerances. For FULL, it dynamically allocates memory and we report its peak memory usage. For MVPipe and our ResidualSketch, we fix the memory usage to $32$KB and $256$KB in 1D-byte and 1D-bit HHH task, respectively. For ResidualSketch, we set two levels on layer $12$ and $32$ and cocosketch (COCO) \cite{cocosketch} as the Residual Block.
\begin{figure*}[htbp] 
\centering
\begin{subfigure}[b]{0.24\textwidth}
\includegraphics[width=\linewidth]{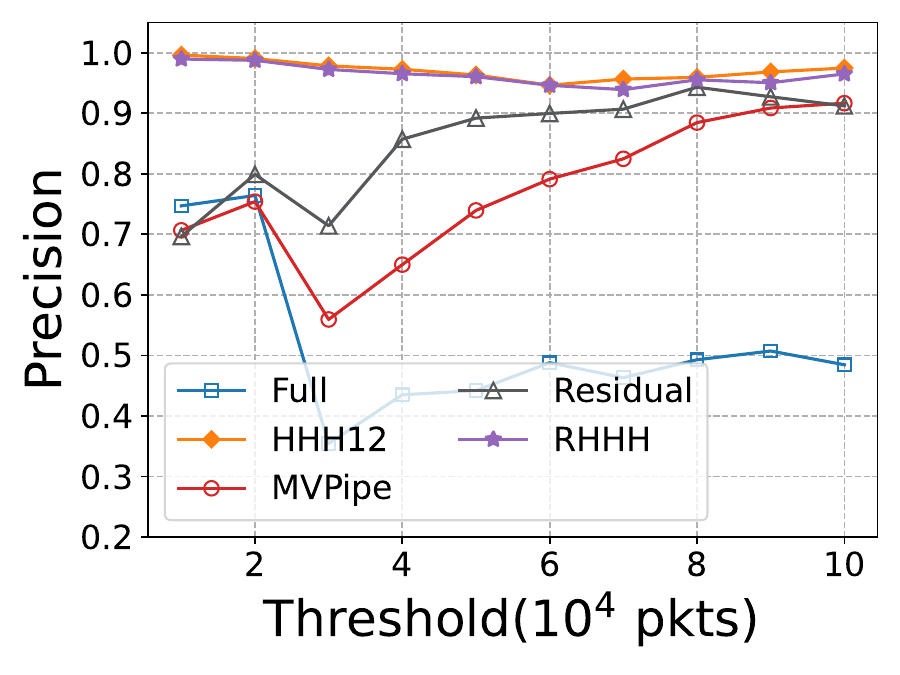}
\caption{Precision for 1D-byte}
\label{fig:mawibytePR}
\end{subfigure}
\hfill 
\begin{subfigure}[b]{0.24\textwidth}
\includegraphics[width=\linewidth]{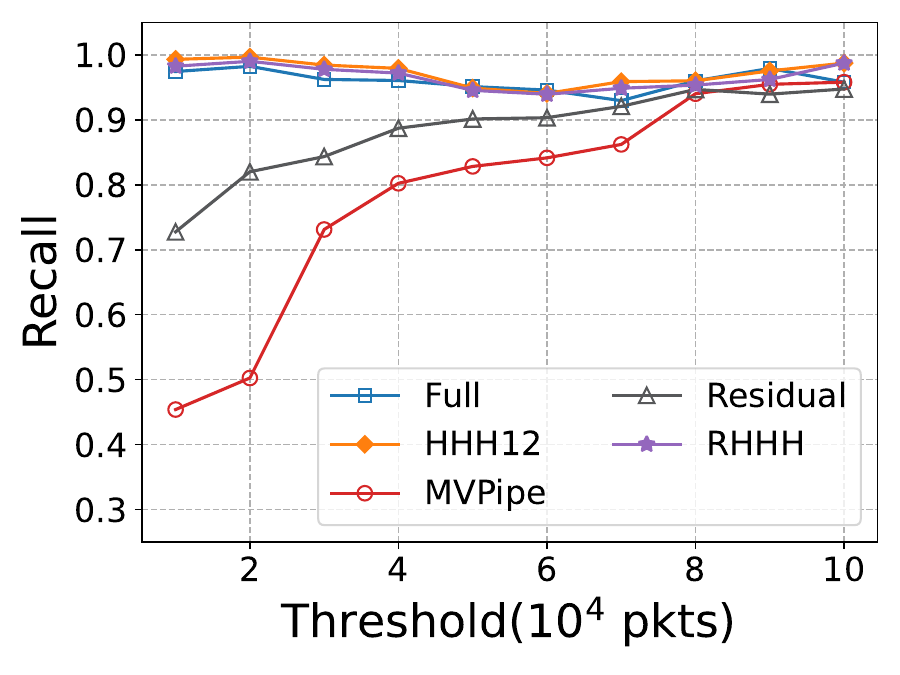}
\caption{Recall for 1D-byte}
\label{fig:mawibyteRR}
\end{subfigure}
\hfill 
\begin{subfigure}[b]{0.24\textwidth}
\includegraphics[width=\linewidth]{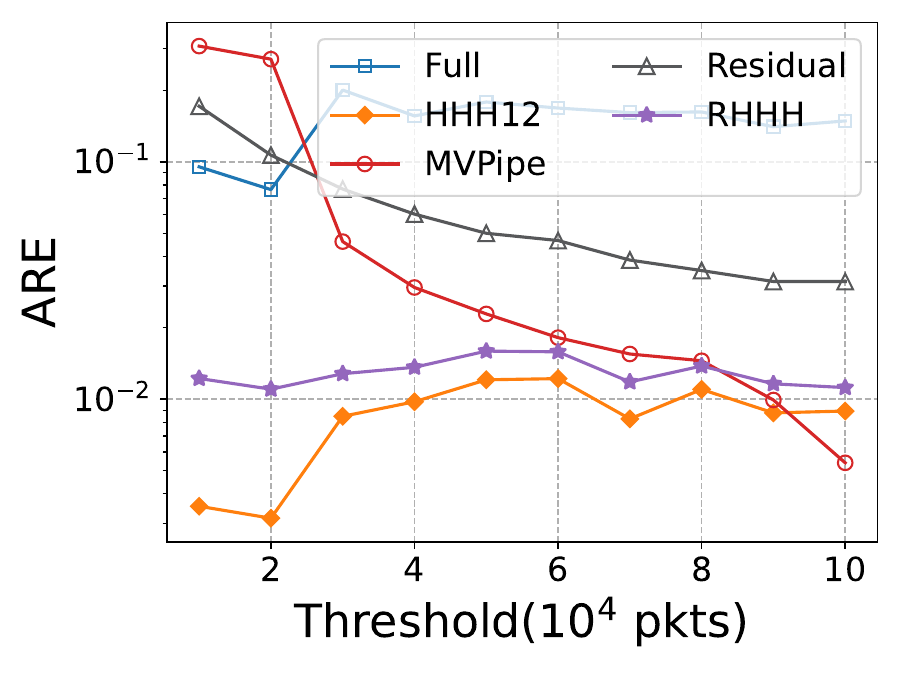}
\caption{ARE for 1D-byte}
\label{fig:mawibyteARE}
\end{subfigure}
\hfill 
\begin{subfigure}[b]{0.24\textwidth}
\includegraphics[width=\linewidth]{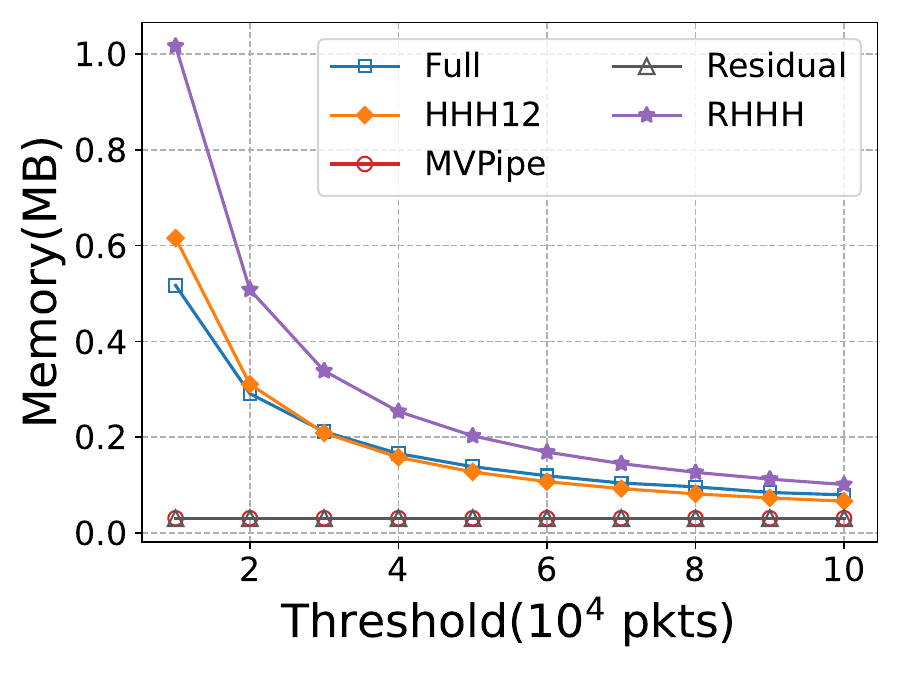}
\caption{Memory Usage for 1D-byte}
\label{fig:mawibyteMEM}
\end{subfigure}
\begin{subfigure}[b]{0.24\textwidth}
\includegraphics[width=\linewidth]{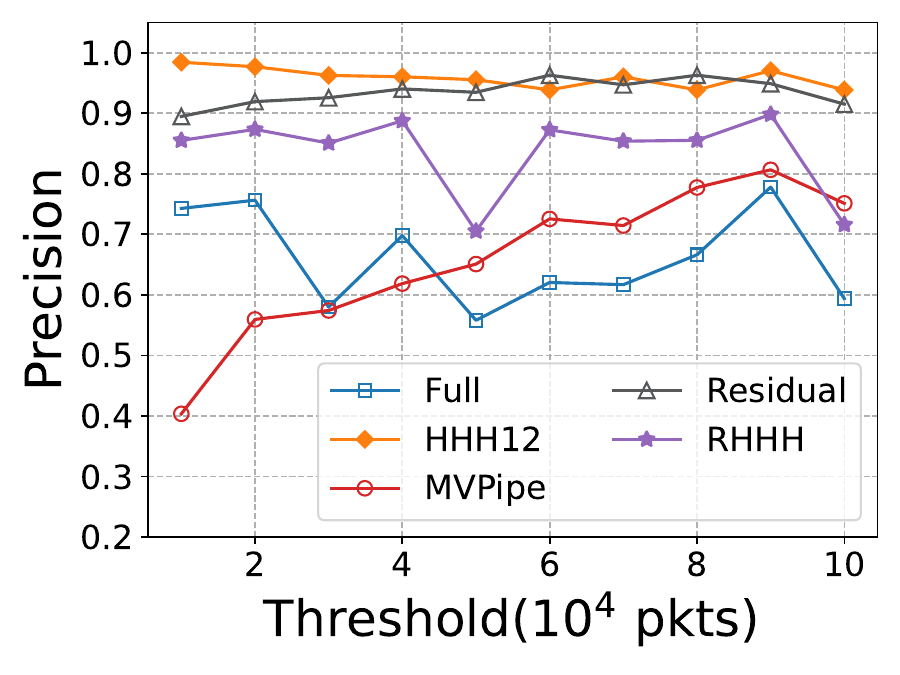}
\caption{Precision for 1D-bit}
\label{fig:mawibitPR}
\end{subfigure}
\hfill 
\begin{subfigure}[b]{0.24\textwidth}
\includegraphics[width=\linewidth]{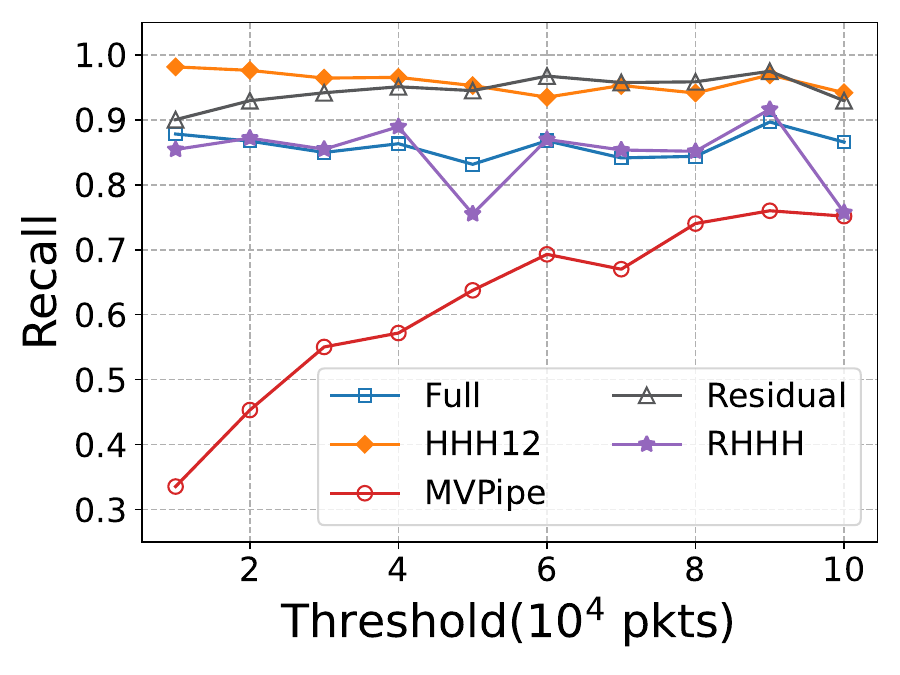}
\caption{Recall for 1D-bit}
\label{fig:mawibitRR}
\end{subfigure}
\hfill 
\begin{subfigure}[b]{0.24\textwidth}
\includegraphics[width=\linewidth]{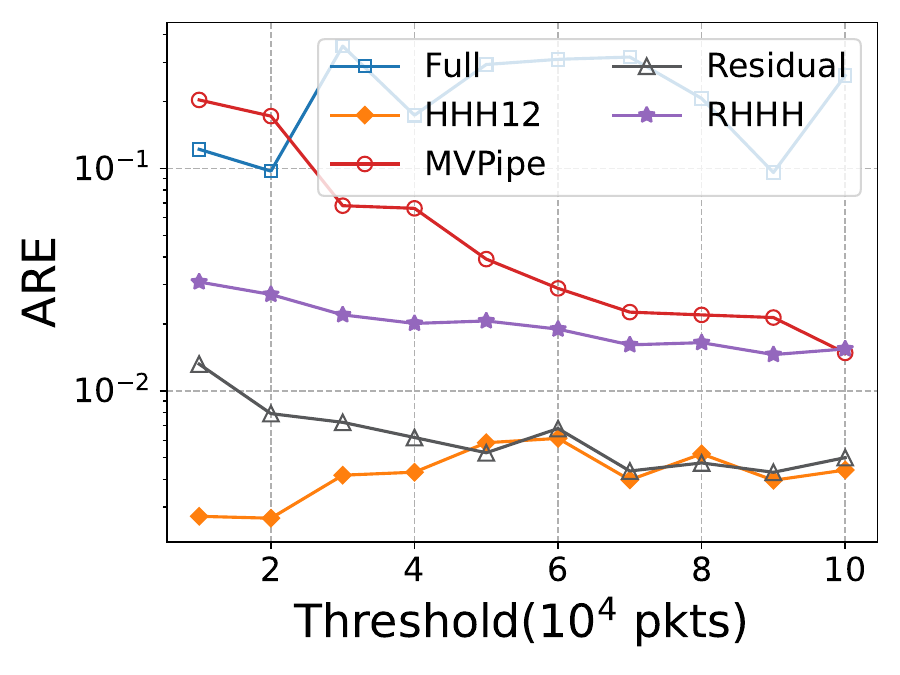}
\caption{ARE for 1D-bit}
\label{fig:mawibitARE}
\end{subfigure}
\hfill 
\begin{subfigure}[b]{0.24\textwidth}
\includegraphics[width=\linewidth]{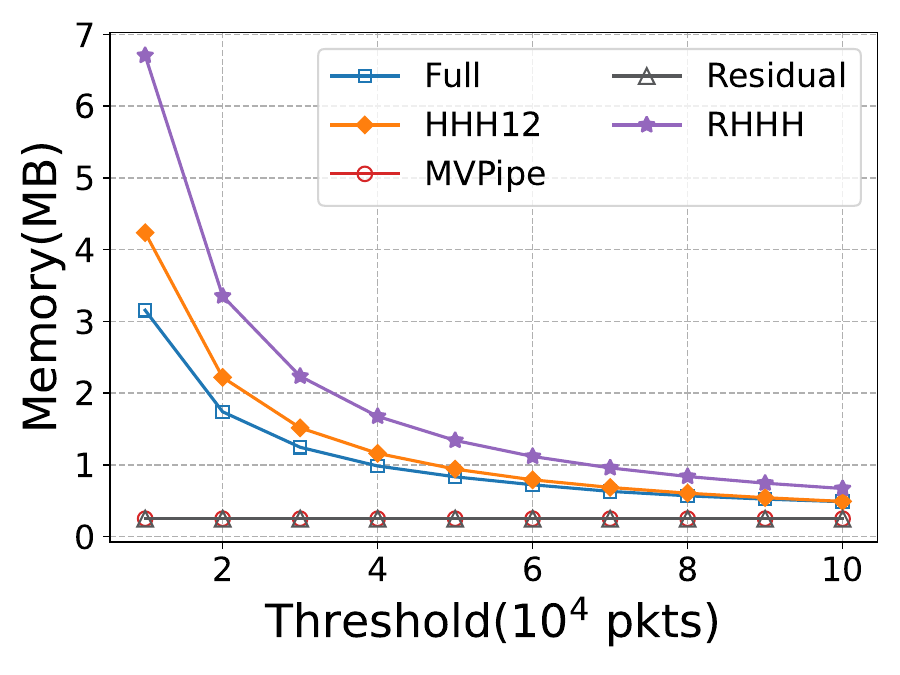}
\caption{Memory Usage for 1D-bit}
\label{fig:mawibitMEM}
\end{subfigure}
\caption{Performance Comparison with HH-based Sketches on MAWI Dataset.}
\label{mawi_exp}
\vspace{-0.1in}
\end{figure*}

\noindent\textbf{PR and RR (Figure \ref{fig:mawibytePR}, \ref{fig:mawibyteRR}, \ref{fig:mawibitPR}, \ref{fig:mawibitRR})}: We find that ResidualSketch maintains high PR and RR, exceeding $90\%$ across various thresholds, even with minimal memory utilization, except for scenarios involving small thresholds within the 1-byte HHH task. Conversely, MVPipe demonstrates comparatively lower accuracy, ranging from $40.3\%$ to $91.6\%$, attributable to its limited memory allocation. For HHH12 and RHHH, the allocation of adequate memory contributes to their attainment of high PR and RR. Meanwhile, FULL is characterized by its tendency to report more HHs than actually exist, resulting in comparatively lower precision. Given that the magnitudes of PR and RR are roughly equivalent except FULL, subsequent experiments will employ the F1 score.

\noindent\textbf{ARE (Figure \ref{fig:mawibyteARE}, \ref{fig:mawibitARE})}: We find that ResidualSketch achieves a low ARE less than $0.01$ in 1D-bit HHH detection, closely approaching the performance of HHH12. However, it does not behave well in 1D-byte case. As there are limited layers in 1D-byte case, other benchmark algorithms have minimal memory waste and achieve better ARE results.

\noindent\textbf{Conclusion}: With minimal memory, ResidualSketch achieves almost best accuracy and lowest ARE in 1D-bit case, though it underperforms in 1D-byte case. With an increased number of layers, HH-based sketches have more memory waste  and the advantages of ResidualSketch become more pronounced. 


\subsubsection{Comparison with APK-based Sketches}
We compare ResidualSketch with USS \cite{unbiasedspacesaving} and Cocosketch (COCO) \cite{cocosketch} for APK-based sketches. Here, our focus is primarily on the performance in detecting HHs aggregated by small flows. Consequently, in our comparison metrics, we exclude the counting of HHs in base layer. For ResidualSketch, we set two levels on layers $12$ and $32$ for CAIDA dataset and three levels on layers $12$, $24$ and $32$ on synthetic dataset. We both use COCO and USS as the Residual Block.

\begin{figure}[htbp] 
\vspace{-0.1in}
\centering
\begin{subfigure}[b]{0.24\textwidth}
\includegraphics[width=\linewidth]{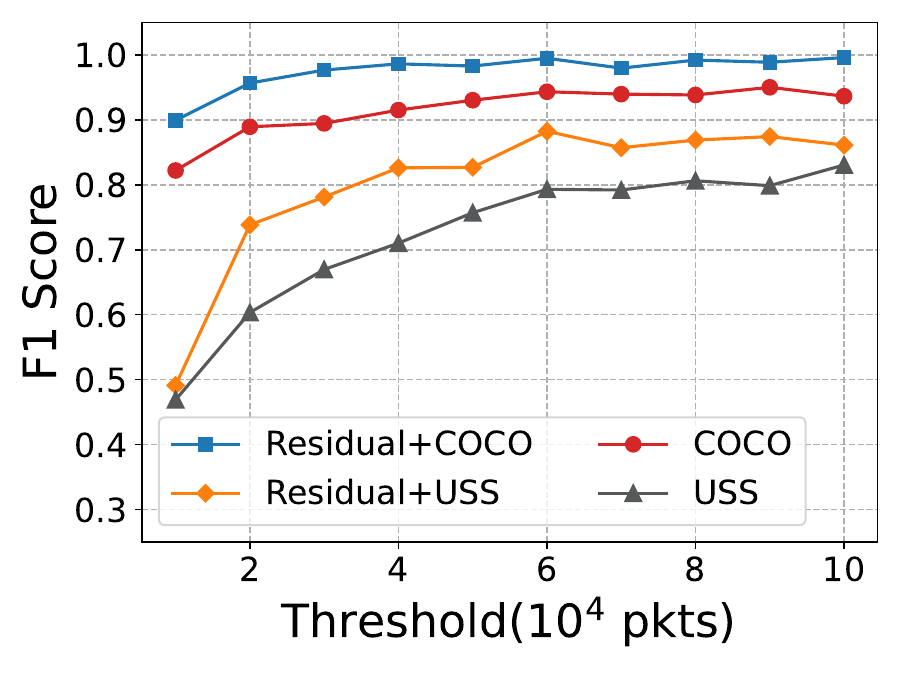}
\caption{F1 Score for 1D-bit}
\label{fig:caidaF1score}
\end{subfigure}
\hfill 
\begin{subfigure}[b]{0.24\textwidth}
\includegraphics[width=\linewidth]{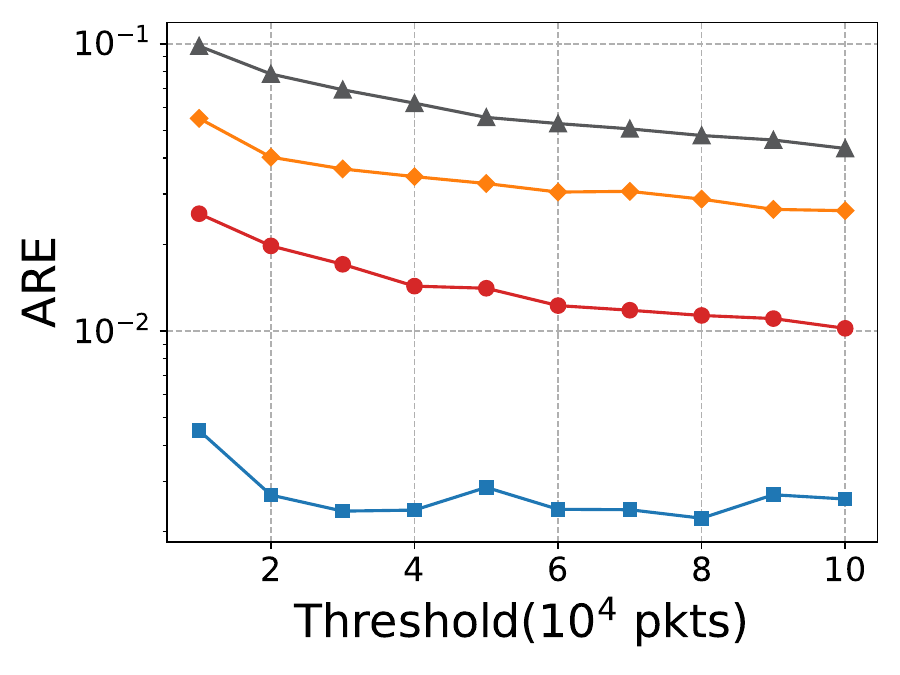}
\caption{ARE for 1D-bit}
\label{fig:caidaARE}
\end{subfigure}
\hfill 
\vspace{-0.1in}
\caption{Performance Comparison with APK-based Sketches on CAIDA Dataset.(\ref{fig:caidaARE}'s legend is the same as \ref{fig:caidaF1score})}
\label{APK_exp}
\vspace{-0.15in}
\end{figure}

\noindent\textbf{Experimental Results on CAIDA (Figure \ref{APK_exp})}: We use $256$KB memory and vary the threshold from $10$k to $100$k. We find that a performance ranking from best to worst is as follows: ResidualSketch with COCO as the Residual Block (Residual+COCO), COCO, Residual+USS, and USS. Notably, the F1 score of Residual+COCO exceeds $90\%$, numerically surpassing COCO by $5\%$.  Residual+COCO's ARE ranges between $0.0022$ and $0.0049$, representing an improvement to one-fifth or even better compared to COCO.

\noindent\textbf{Experimental Results on Synthetic (Figure \ref{APK_exp_synthetic})}: Utilizing $512$KB of memory, a $20$k threshold, and varying the top-1000 ratio from $10\%$ to $50\%$, Residual+COCO stands out as the most effective, achieving an F1 score exceeding $87\%$ and an ARE of less than $0.005$. This represents a $26\%$ improvement in F1 score over COCO, with the ARE being one-tenth of COCO's. As the top-1000 ratio increases, the detection outcomes for all algorithms show enhancement. Remarkably, the performance trend suggests that Residual+USS is comparable with COCO, even though USS employs a StreamSummary structure that requires more space due to its hash table and doubly linked list components. 

\begin{figure}[htbp] 
\vspace{-0.1in}
\centering
\begin{subfigure}[b]{0.24\textwidth}
\includegraphics[width=\linewidth]{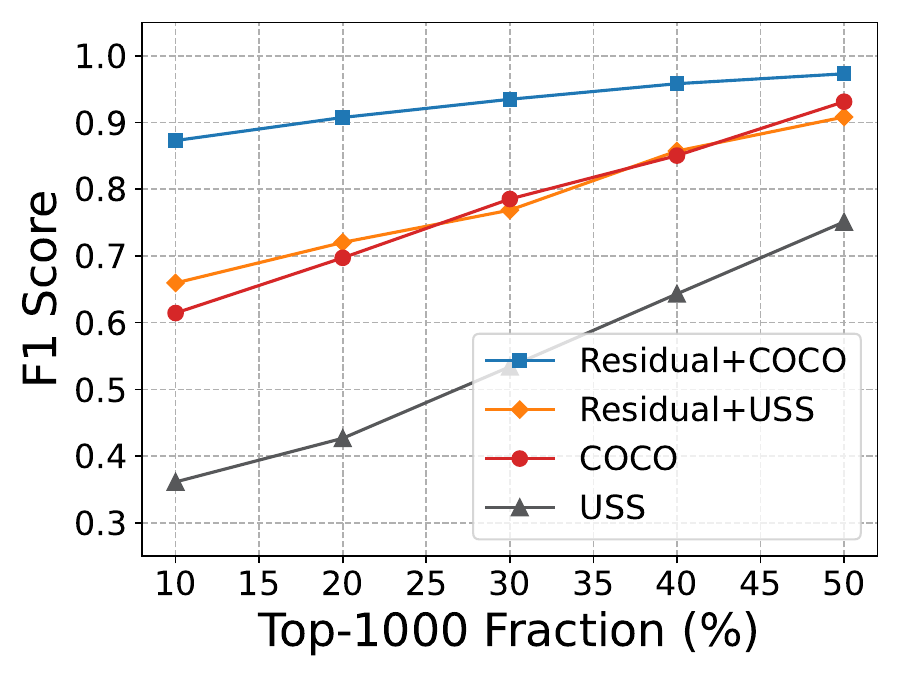}
\caption{F1 Score for 1D-bit}
\label{fig:synF1score}
\end{subfigure}
\hfill 
\begin{subfigure}[b]{0.24\textwidth}
\includegraphics[width=\linewidth]{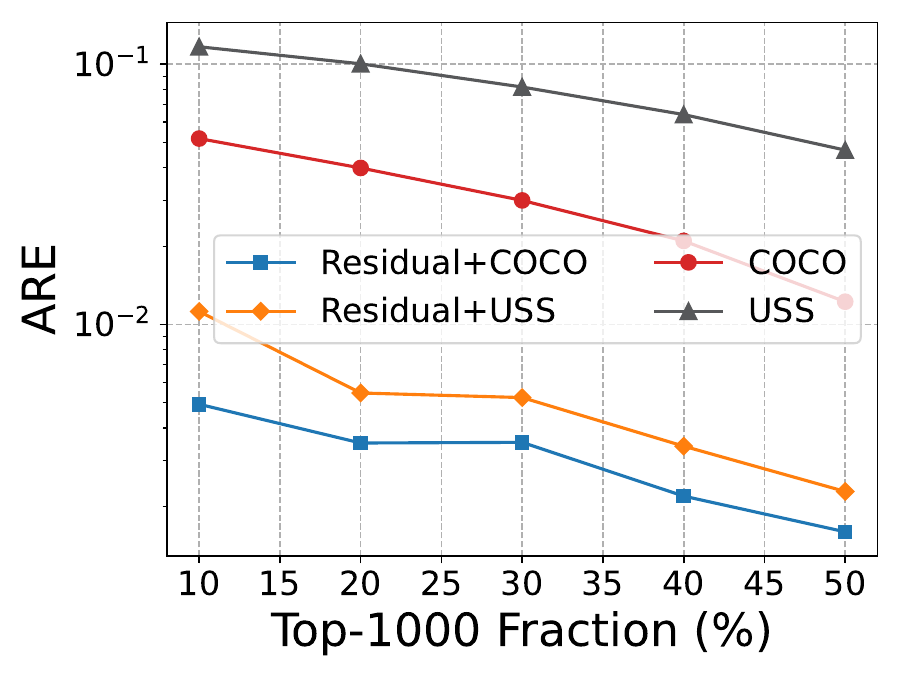}
\caption{ARE for 1D-bit}
\label{fig:synARE}
\end{subfigure}
\hfill 
\vspace{-0.1in}
\caption{Performance Comparison with APK-based Sketches on Synthetic Dataset.}
\label{APK_exp_synthetic}
\vspace{-0.15in}
\end{figure}

\begin{figure}[htbp] 
\centering
\begin{subfigure}[b]{0.48\textwidth}
\includegraphics[width=\linewidth]{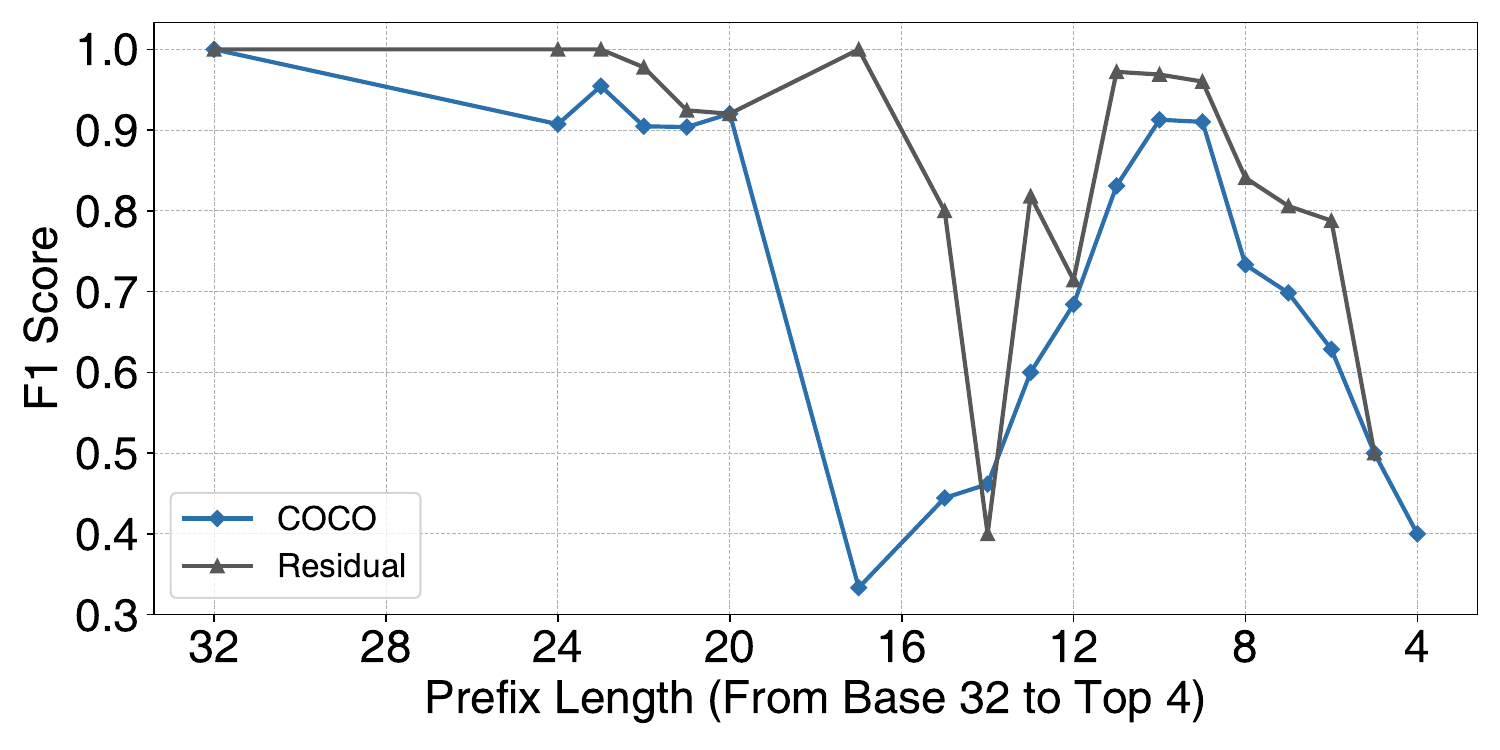}
\label{fig:GED}
\end{subfigure}
\vspace{-0.1in}
\caption{F1 Score for HHH Detection Across Various Layers in the Synthetic Dataset  }
\label{APK_exp_GED}
\vspace{-0.1in}
\end{figure}

\noindent\textbf{Experimental Results on GED (Figure \ref{APK_exp_synthetic})}: In the experiments conducted on the synthetic dataset, we re-examine the phenomenon of GED underlying the COCO algorithm to understand the improved results. We specifically compare the F1 scores of the Residual+COCO relative to the COCO method across different layers. At a threshold of $20$k, memory with $256$KB and a top-1000 ratio of $30\%$, it's observed that the COCO method exhibited a significant decrease in F1 score at higher layers. In contrast, the Residual+COCO method, which reconstructs the data structure at specific intervals, exhibits lower F1 scores only at the 5th and 14th layer, a decrease due to the inaccurate estimation of few isolated HHs at these layers.

\noindent\textbf{Conclusion}:
Our algorithm, building upon the existing APK-based sketches, COCO and USS, demonstrates significant improvements in both F1 score and ARE across both the CAIDA dataset and synthetic datasets. We also analyzed how the application of our strategy mitigated the GED phenomenon in the original COCO, thereby validating the effectiveness of our approach.

\begin{figure}[htbp] 
\vspace{-0.15in}
\centering
\begin{subfigure}[b]{0.24\textwidth}
\includegraphics[width=\linewidth]{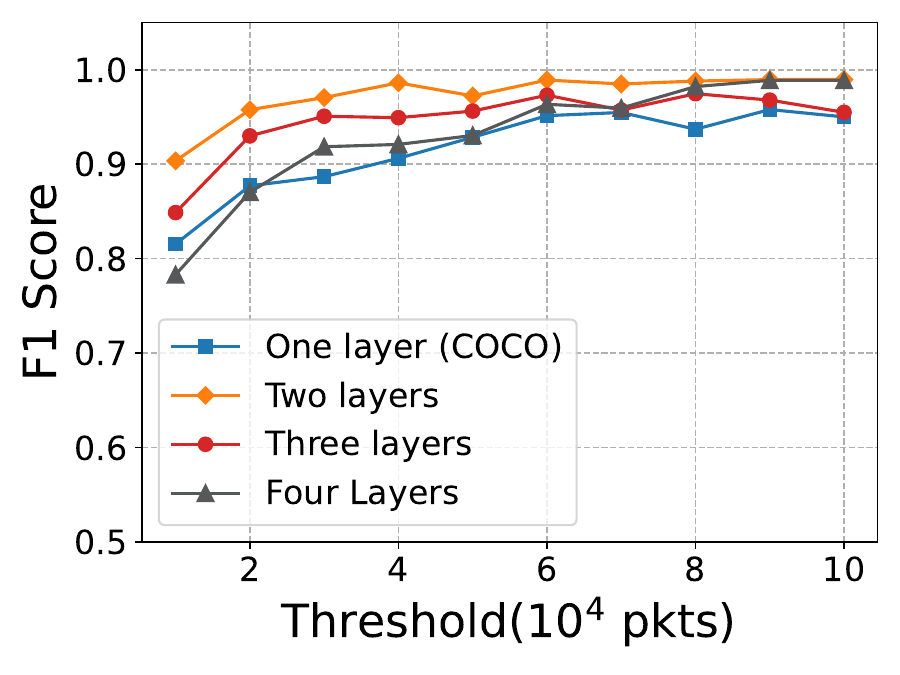}
\label{fig:pmemory_layer_F1_score}
\end{subfigure}
\hfill
\begin{subfigure}[b]{0.24\textwidth}
\includegraphics[width=\linewidth]{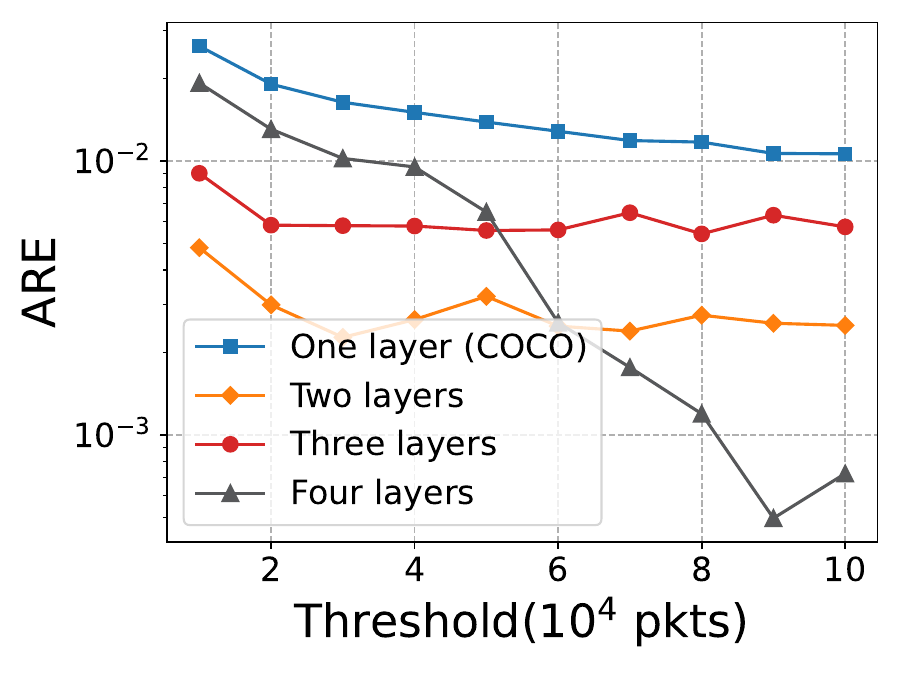}
\label{fig:pmemory_layer_ARE}
\end{subfigure}
\vspace{-0.15in}
\caption{Evaluation of ResidualSketch on Parameter Setting: Layer Number}
\label{parameter_layer}
\vspace{-0.1in}
\end{figure}

\begin{figure}[htbp] 
\vspace{-0.15in}
\centering
\begin{subfigure}[b]{0.24\textwidth}
\includegraphics[width=\linewidth]{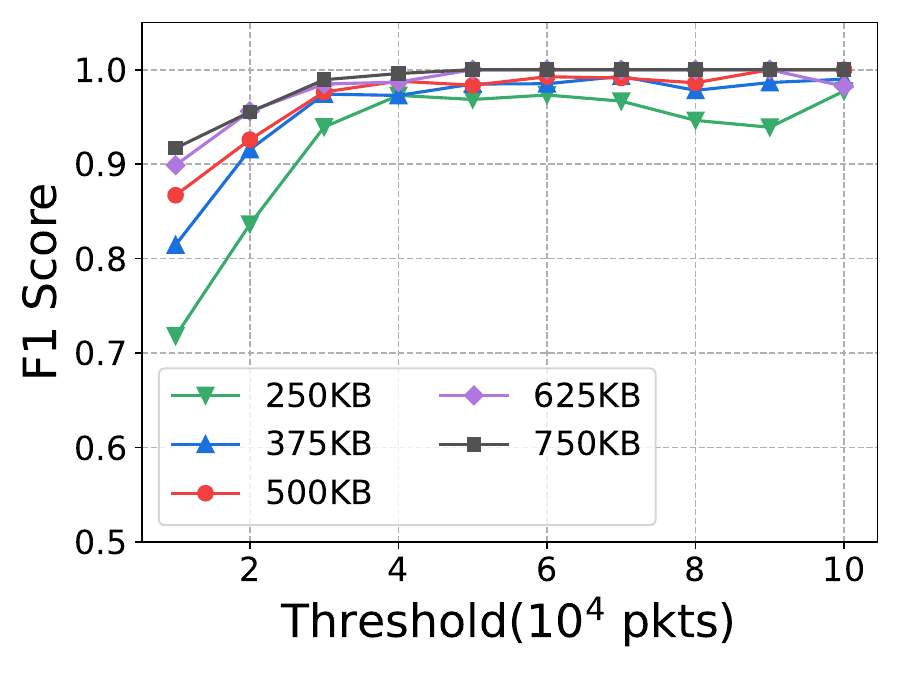}
\label{fig:pmemory_F1_score}
\end{subfigure}
\hfill
\begin{subfigure}[b]{0.24\textwidth}
\includegraphics[width=\linewidth]{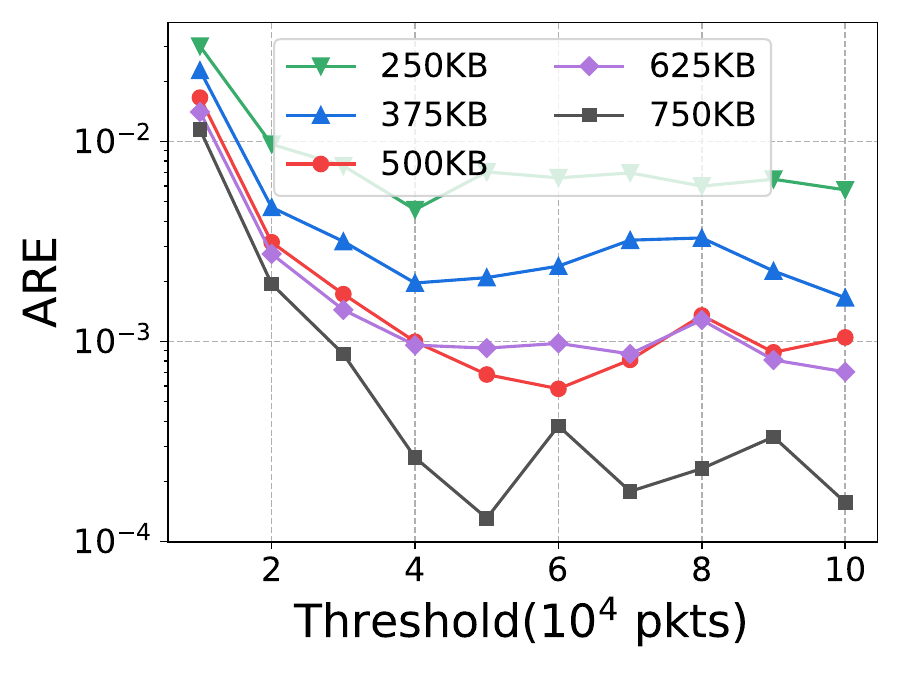}
\label{fig:pmemory_ARE}
\end{subfigure}

\caption{Evaluation of ResidualSketch on Parameter Setting: Memory Size}
\label{parameter_memory}
\vspace{-0.1in}
\end{figure}

\vspace{-0.05in}
\subsection{The impact of Algorithm Parameters}

\noindent\textbf{Experimental results on parameter $L$ (Figure \ref{parameter_layer})}: In this experiment, we vary the number of layers in the CAIDA dataset to observe the impact of layer configuration on performance. We allocate a memory size of 256KB, with two layers set to $12,32$, three layers to $12,24,$ and $32$, consistent with previous experiments. Additionally, four layers are configured as $8, 14, 24, 32$. We discovered that the optimal results were achieved with two layers; the results for three layers were the next best, still outperforming the single-layer configuration, i.e., COCO. With four layers, at lower thresholds, both F1 score and ARE are similar to those observed with a single layer. This outcome suggests that when the number of layers increases, the allocated memory per layer decreases, leading to a scenario where memory constraints, rather than GED, start to dominate performance.

\noindent\textbf{Experimental Results on parameter memory size (Figure \ref{parameter_memory})}: In the synthetic dataset with a top-1000 ratio of $30\%$, we observe that the results for ResidualSketch are relatively robust when adjusting the memory size. Specifically, at a lower threshold of 10k, the F1 score for a memory size of 250KB still reaches $72\%$. At other thresholds, the F1 scores for different memory sizes all exceed $90\%$, with the ARE below $0.01$. 

\begin{figure}[htbp] 
\vspace{-0.1in}
\centering
\begin{subfigure}[b]{0.48\textwidth}
\includegraphics[width=\linewidth]{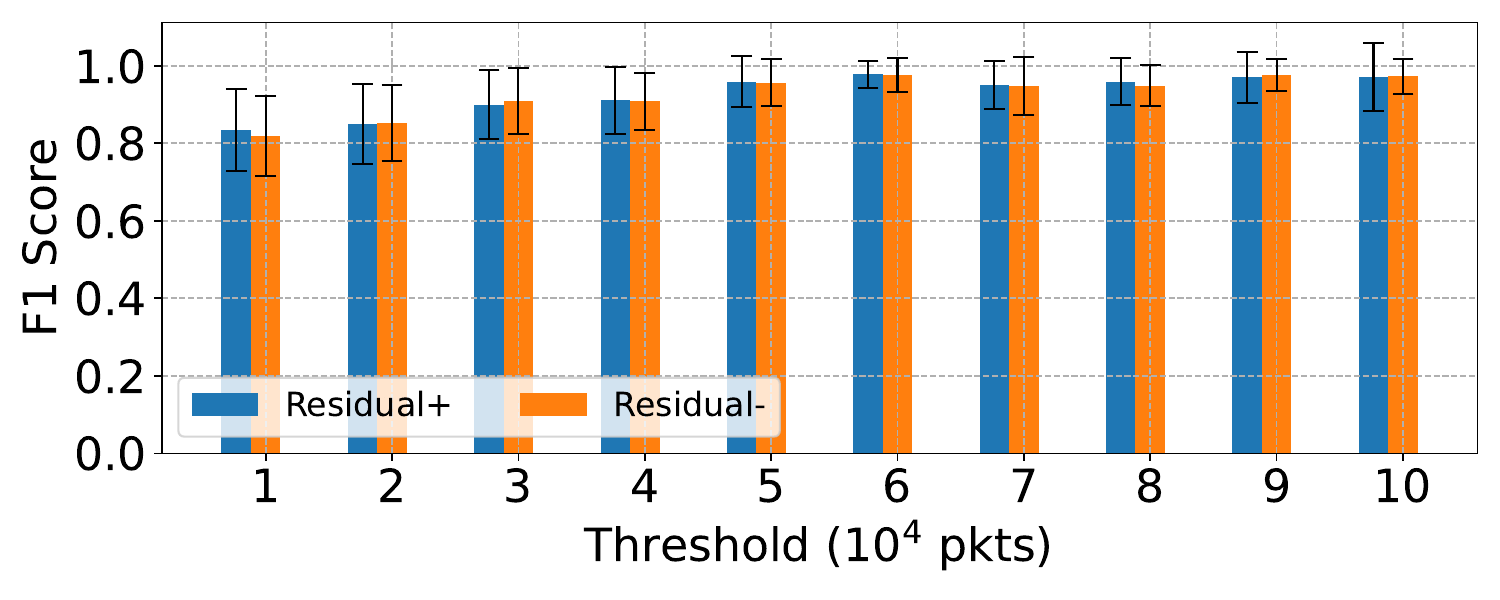}
\label{fig:connection}
\end{subfigure}
\vspace{-0.1in}
\caption{Evaluation of ResidualSketch on Parameter Setting: Residual Connection}
\label{exp_connection}
\vspace{-0.1in}
\end{figure}

\begin{figure}[htbp] 
\vspace{-0.05in}
\centering
\begin{subfigure}[b]{0.48\textwidth}
\includegraphics[width=\linewidth]{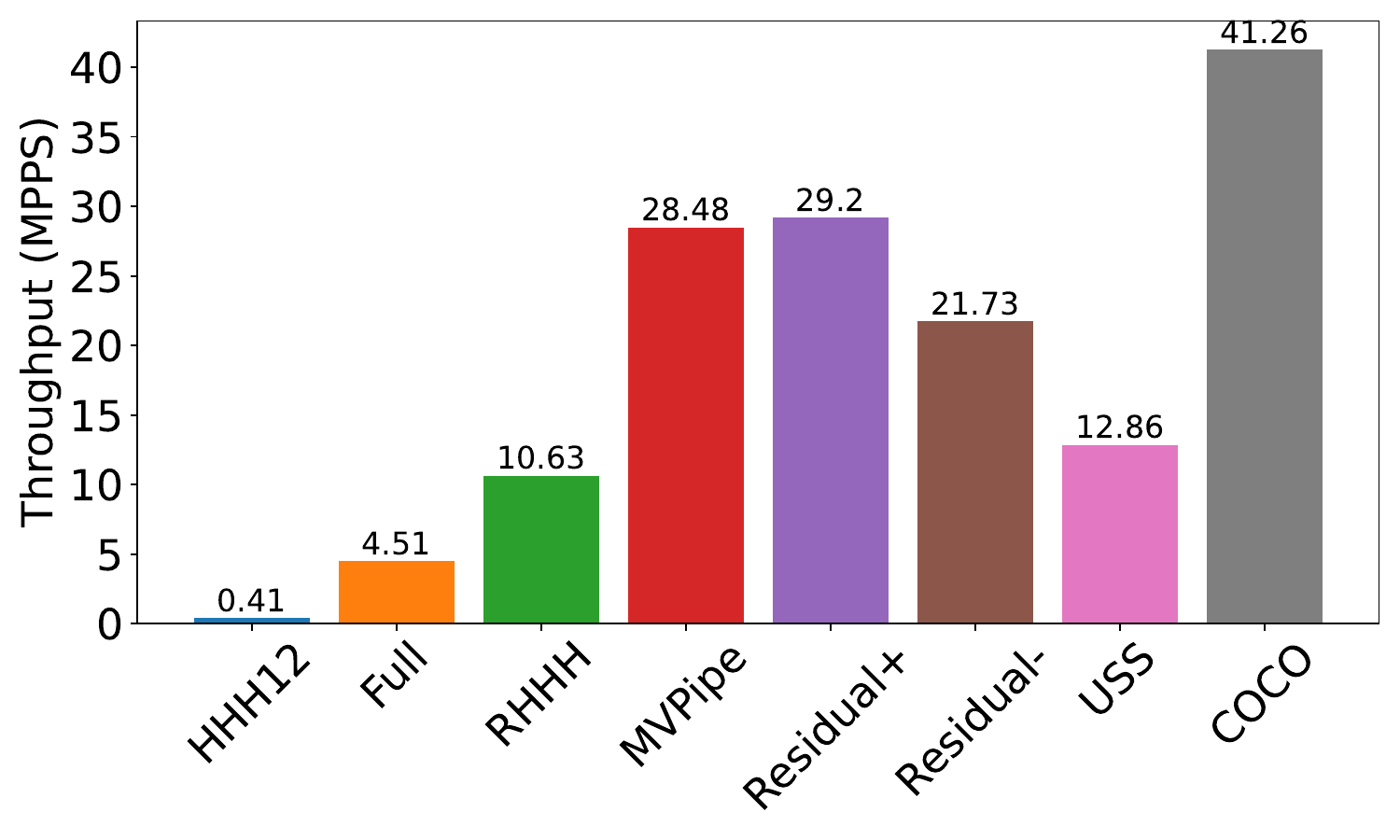}
\label{fig:throughput}
\end{subfigure}
\vspace{-0.1in}
\caption{Processing Speed Comparison on Campus Dataset, including ResidualSketch without Residual Connection}
\label{exp_throughput}
\vspace{-0.1in}
\end{figure}

\noindent\textbf{Experimental Results on Residual Connection (Figure \ref{exp_connection}, \ref{exp_throughput})}: Utilizing a campus dataset with a threshold of $20$k, we assess the accuracy of ResidualSketch both with and without the Residual Connection. The findings indicate comparable accuracy levels, attributed to small memory release from decrement. Notably, the Residual Connection strategy enhances throughput from $21.73$ to $29.2$ Mpps, marking a $34.3\%$ increase. In comparison, while COCO achieves $41.26$ Mpps, MVPipe's throughput is close to ours at 28.48 Mpps.

\noindent\textbf{Conclusion}: We recommend adopting a configuration of either two or three layers, coupled with residual connections, to enhance both accuracy and throughput.

\presec
\section{Conclusion}
\label{sec:conclusion}
\postsec
In this study, we introduce a novel approach for the processing of Hierarchical Heavy Hitters (HHH) in data streams, a task critical for the swift and precise identification of traffic anomalies. Our proposed method, the ResidualSketch, leverages the innovative application of \textit{Residual Blocks} and \textit{Residual Connection} techniques, drawing inspiration from ResNet architectures. This method strategically utilizes Residual Blocks at pivotal layers within the IP hierarchy to reset error propagation and effectively counteract the \textit{Gradual Error Diffusion} that affects existing sketches, while also enhancing throughput through Residual Connections. Through rigorous theoretical analysis and comprehensive experimental evaluations, we demonstrate that ResidualSketch outperforms current state-of-the-art solutions in terms of layer efficiency and error reduction, substantiating its superiority and effectiveness in HHH detection.

\section*{Acknowledgement}
This work was supported by Beijing Natural Natural Science
Foundation (Grant No. QY23043). The corresponding authors are Gaogang Xie and Tong Yang.

\bibliographystyle{unsrt}
\bibliography{reference.bib}

\clearpage
\end{document}